\documentclass[11pt,reqno]{amsart}
\usepackage[utf8]{inputenc}
\usepackage{lmodern} 
\usepackage{amsmath,amssymb,amsfonts,amsthm,latexsym}
\usepackage{graphicx,multirow,enumerate}
\usepackage{graphicx}
\usepackage{tikz}
\usetikzlibrary{positioning, arrows.meta}
\usepackage{booktabs} 
\usepackage{float}
\usepackage [numbers,sort&compress]{natbib}
\usepackage{caption}
\usepackage{array}      
\usepackage{cellspace}  
\usepackage{mathrsfs}   
\usepackage[all]{xy}    
\usepackage{color,xcolor} 
\usepackage{cancel} 
\usepackage{url}        
\usepackage{longtable}
\usepackage{hyperref}
\theoremstyle{plain}
\newtheorem{thm}{Theorem}[section]
\newtheorem{lem}[thm]{Lemma}
\newtheorem{prop}[thm]{Proposition}
\newtheorem{cor}[thm]{Corollary}
\theoremstyle{definition}
\newtheorem{ex}[thm]{Example}
\newtheorem{defn}[thm]{Definition}
\newtheorem{rem}[thm]{Remark}

\definecolor{lightyellow}{rgb}{1, 1, 0.8} 
\usepackage{float}
\usepackage{placeins}
\usepackage{graphicx}

\begin{document}
\title[]{ $A$-Generalized Hessian pre-Lie algebras and $A$-Generalized Yang--Baxter Equations}
\author{Yining Sun}
	\address{Y. Sun: School of
		Mathematics and Statistics, Northeast Normal
		University, Changchun 130024, China}
	\email{yiningsun@nenu.edu.cn}
	
	\author{Zeyu Hao}
	\address{Z. Hao: School of
		Mathematics and Statistics, Northeast Normal
		University, Changchun 130024, China}
	\email{haozeyu@nenu.edu.cn}
	
	\author{Ziyi Zhang}
	\address{Z. Zhang: School of
		Mathematics and Statistics, Northeast Normal
		University, Changchun 130024, China}
	\email{zhangziyi@nenu.edu.cn}
	
	\author{Liangyun Chen$^{*}$}
	\address{L. Chen: School of
		Mathematics and Statistics, Northeast Normal
		University, Changchun 130024, China}
	\email{chenly640@nenu.edu.cn}
\begin{abstract}
Inspired by the problem of constructing ($\omega$-)pre-Lie algebra structures on the dual space of a pre-Lie algebra, we introduce the \(A\)-generalized Yang--Baxter equation as a generalization of the Yang--Baxter equation of pre-Lie algebras. We study its symmetric solutions through
\(A\)-generalized Hessian pre-Lie algebras and split these solutions
into two types. We further consider factorizable solutions of this equation and establish a one-to-one correspondence between them and generalized quadratic Rota--Baxter pre-Lie algebras of nonzero weight. By studying the structure of these algebras, we find all factorizable solutions. Finally, we study the structure of \(A\)-generalized Hessian pre-Lie algebras. In particular, we obtain a structural description via central and double extensions and classify low-dimensional non-trivial \(A\)-generalized Hessian pre-Lie algebras.

        
       \textit{} 
		
	   \end{abstract}
	
	\thanks{*Corresponding author.}
	
	\thanks{This work is supported by NNSF of China (No. 12271085)}

	\maketitle

    \section{Introduction}
In \cite{Tah1979}, the classical Yang--Baxter equation was introduced
in connection with inverse scattering theory. Subsequently, it has played
an important role in many areas, including integrable systems, quantum
groups, and quantum field theory \cite{Etingof1998}. Since Baxter's work \cite{Baxter1960} and Rota's subsequent development \cite{Rota1969}, Rota--Baxter operators have been widely studied and have found important applications in mathematical physics \cite{EbrahimiFardGuo2007}, number theory \cite{EbrahimiFardGuo2008}, and noncommutative geometry \cite{ConnesKreimer1998}. In recent years, Yang--Baxter equations, Rota--Baxter operators, and their various generalizations have been studied extensively. The Yang--Baxter equation of pre-Lie algebras appears naturally in coboundary pre-Lie bialgebras and serves as the pre-Lie analogue of the classical Yang--Baxter equation \cite{Bai2008}. 

\(\omega\)-Lie algebras were introduced by Nurowski as deformations of Lie
algebras by means of a two-form \cite{Nurowski2007}, and their structure theory was developed by Zusmanovich \cite{Zusmanovich2010}. Subsequent work on \(\omega\)-Lie (super)algebras has investigated their low-dimensional classifications \cite{ChenLiuZhang2014,ChenZhang2017,ChenNi2023,ChenL2021,ChenL2018},
derivations \cite{ChenZhangZhangZhuang2018,Oubba2024,ChenRenShanZhang2025GD,Oubba2025,J. Zhou250522}, representations \cite{ChenL2018,R. Zhang2021}, and Rota--Baxter operators \cite{ChenRenShanZhang2026RB}. The corresponding notion of an \(\omega\)-pre-Lie algebra was introduced by Chen \cite{ChenWu2023}. Its commutator naturally defines an \(\omega\)-Lie algebra. Recent developments include low-dimensional classifications \cite{ChenNiYu2026FourDim,ChenWu2023}and structural results \cite{ChenNiYu2025Structure}.

How to equip the dual space of a given algebra with the same type of
algebraic structure is a classical question. It is well known that
skew-symmetric solutions of the classical Yang--Baxter equation of a Lie
algebra induce Lie algebra structures on its dual space. Similarly,
symmetric solutions of the Yang--Baxter equation of a pre-Lie algebra induce pre-Lie algebra structures on its dual space. In fact, (skew-)symmetric solutions of analogues of the classical Yang--Baxter equation often give the corresponding algebraic structures on dual spaces, for instance in the case of Leibniz algebras \cite{Barreiro2016,Tang2022},3-Lie algebras \cite{BaiGuoSheng2019} and conformal algebras \cite{HongBai2020} and so on. 

We observe that the construction of pre-Lie algebra structures on dual spaces is not restricted to solutions of the classical Yang--Baxter
equation. Based on this observation, we introduce the \(A\)-generalized
Yang--Baxter equation. More precisely, for a fixed element \(u\in \operatorname{Ann}_{(\mathrm R)}(A)\), we define the \(u\)-generalized Yang--Baxter equation and prove that its solutions induce ($\omega$-)pre-Lie algebra structures on the dual spaces. 

One of the central aims of this paper is to construct solutions
of the \(A\)-generalized Yang--Baxter equation and explain their algebraic and geometric meanings. For an arbitrary pre-Lie algebra, symmetric
solutions are characterized by \(A\)-generalized relative Rota--Baxter
operators associated with the coregular representation. This gives a generalized version of Bai's result \cite{Bai2008}. We also introduce \(A\)-generalized Hessian pre-Lie algebras and prove that they give exactly the nondegenerate symmetric solutions. This yields a division of symmetric solutions into two classes: type 1 solutions correspond to \(A\)-generalized Hessian pre-Lie subalgebras, while type 2 solutions are described by ordinary Hessian pre-Lie algebras together with an annihilator element.
 
We also study non-symmetric solutions, especially factorizable solutions.
Such solutions are connected with factorizable pre-Lie bialgebras
and Rota--Baxter operators of nonzero weight \cite{WangBaiLiuSheng2024}. In recent years, factorizable solutions have been widely studied for
several algebraic structures, including Lie algebras
\cite{Goncharov2017,Goncharov2019,Lang2023}, Leibniz algebras
\cite{LiuGuilai2025}, antisymmetric infinitesimal algebras
\cite{ShengWang2023factorizable}, and related algebraic structures. We construct and interpret factorizable solutions of the \(A\)-generalized Yang--Baxter equation in terms of \(A\)-generalized quadratic Rota--Baxter pre-Lie algebras of nonzero weight. Moreover, we establish a one-to-one correspondence between \(A\)-generalized quadratic Rota--Baxter pre-Lie algebras of nonzero weight and \(A\)-generalized Rota--Baxter symplectic Lie algebras of nonzero weight. By studying the structure of the latter, we show that such non-trivial algebraic structures exist only on two-dimensional abelian Lie algebras. Consequently, all factorizable solutions are obtained.

Finally, we study \(A\)-generalized Hessian pre-Lie algebras in detail. We observe that these algebras have a natural geometric background
\cite{Osipov2023}. Using central extension and double extension techniques, we give a complete structural description of them and show that they are precisely certain extensions of classical Hessian pre-Lie algebras. Moreover,
we classify low-dimensional \(A\)-generalized Hessian left-symmetric
algebras.

To clarify the organization of the paper, we include the following diagrams,
which summarize the main components of our work and their interrelations.

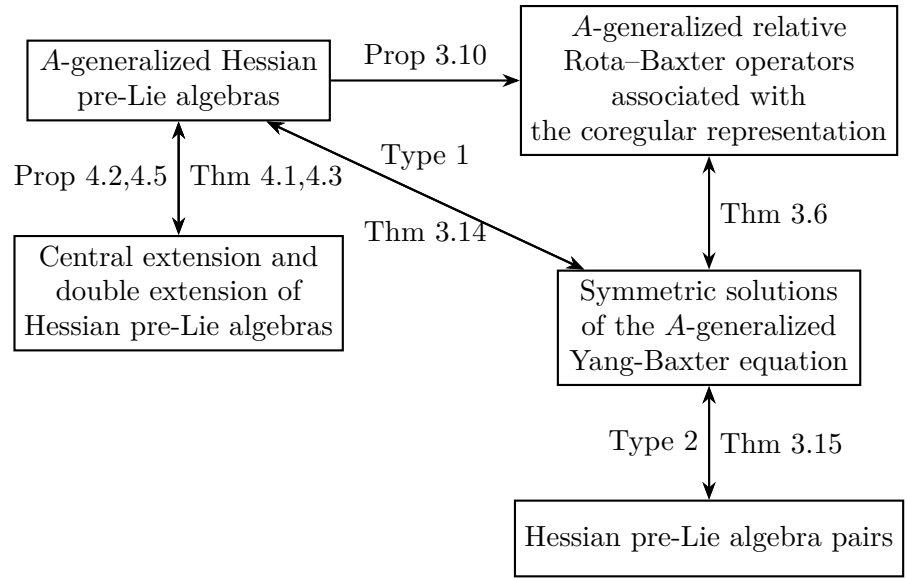
\begin{figure}[H]
    \centering
    \begin{tikzpicture}[
        node distance=1.5cm and 2.5cm,
        box/.style={
            draw=black,       
            rectangle,        
            thick,            
            minimum width=4cm, 
            minimum height=1cm, 
            align=center,     
            font=\normalsize  
        },
        arrow/.style={
            ->,               
            >=Stealth,        
            thick             
        },
        darrow/.style={
            <->,             
            >=Stealth,
            thick
        }
    ]

    \node[box] (n3) {\(A\)-generalized relative\\
Rota--Baxter operators \\associated with\\ the coregular representation};
    
    \node[box, below=of n3] (n5) {Symmetric solutions \\ of the $A$-generalized \\ Yang-Baxter equation};

    \node[box, below=of n5] (n10) {Hessian pre-Lie algebra pairs};

    \node[box, left=of n3] (n2) {\(A\)-generalized Hessian \\ pre-Lie algebras};
    
    \node[box, below=of n2] (n1) {Central extension and \\double extension of \\Hessian pre-Lie algebras};

    \draw[darrow] (n1) -- node[right] {Thm \ref{thm:nonisotropic-u-construction},\ref{thm:isotropic-u-construction}} (n2);

    \draw[darrow] (n1) -- node[left] {Prop \ref{thm:nonisotropic-u-characterization},\ref{thm:isotropic-u-characterization}} (n2);

    \draw[arrow] (n2) -- node[above] {Prop \ref{prop:generalized-Hessian-O-operator}} (n3);
    
    \draw[darrow] (n3) -- node[right] {Thm \ref{thm:generalized-O-operator-S-equation}} (n5);

     \draw[darrow] (n5) -- node[right] {Thm \ref{prop:type2-LSA}} (n10);

     \draw[darrow] (n5) -- node[left] {Type \(2\)} (n10);

    \draw[darrow] (n5) -- node[below=0.2] {Thm \ref{prop:type1-LSA}} (n2);
    
    \draw[darrow] (n5) -- node[above=0.2] {Type \(1\)} (n2);

    \end{tikzpicture}
\end{figure}

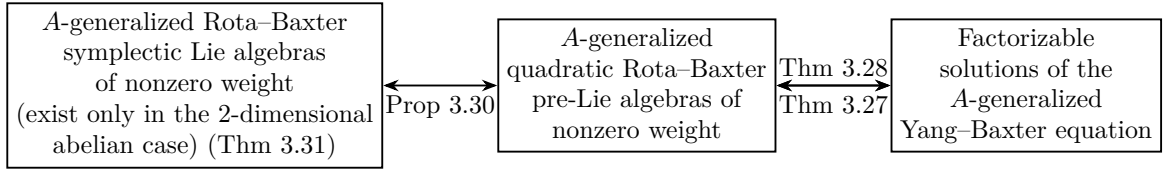
\begin{figure}[H]
    \centering
    \resizebox{0.98\textwidth}{!}{%
    \begin{tikzpicture}[
        node distance=1cm and 1.7cm,
        box/.style={
            draw=black,
            rectangle,
            thick,
            minimum width=3cm,
            minimum height=1cm,
            align=center,
            font=\normalsize
        },
        arrow/.style={
            ->,
            >=Stealth,
            thick
        },
        darrow/.style={
            <->,
            >=Stealth,
            thick
        }
    ]

    \node[box] (n3) {Factorizable \\solutions of the\\ \(A\)-generalized \\ Yang--Baxter equation};

    \node[box, left=of n3] (n2) {\(A\)-generalized \\quadratic
    Rota--Baxter\\ pre-Lie algebras of\\ nonzero weight};

    \node[box, left=of n2] (n4) {\(A\)-generalized Rota--Baxter\\
    symplectic Lie algebras\\
    of nonzero weight\\
    (exist only in the \(2\)-dimensional\\
    abelian case)
    (Thm~\ref{two-dimensionalabelian})};

    \draw[darrow] (n2) -- node[below] {Thm~\ref{thm:generalized-quadratic-RB-to-u-factorizable}} (n3);

    \draw[darrow] (n2) -- node[above] {Thm~\ref{thm:u-factorizable-to-generalized-quadratic-RB}} (n3);

    \draw[darrow] (n2) -- node[below] {Prop~\ref{thm:u-QRB-LSA-symplectic-Lie}} (n4);

    \end{tikzpicture}%
    }
\end{figure}

\FloatBarrier

The paper is organized as follows. In Section~\ref{22}, we recall how solutions of the classical Yang--Baxter equation give pre-Lie algebra
structures on the dual space of a pre-Lie algebra. We then introduce
the \(A\)-generalized Yang--Baxter equation and prove that its solutions can
similarly give \((\omega\)-)pre-Lie algebra structures on the dual space. In Section~\ref{33}, we study solutions of the \(A\)-generalized Yang--Baxter equation. We first characterize symmetric solutions on arbitrary
pre-Lie algebras in terms of \(A\)-generalized relative Rota--Baxter operators associated with the coregular representation. We then introduce \(A\)-generalized Hessian pre-Lie algebras and prove their equivalence with nondegenerate symmetric solutions. On this basis, symmetric solutions are divided into two types: type 1 solutions correspond to \(A\)-generalized Hessian pre-Lie algebras, whereas type 2 solutions are characterized by Hessian pre-Lie algebra pairs. In addition, we study non-trivial factorizable solutions and characterize them in terms of a concrete finite class of \(A\)-generalized quadratic Rota--Baxter pre-Lie algebras of nonzero weight. This gives all factorizable solutions. In Section~\ref{44}, we study \(A\)-generalized Hessian pre-Lie algebras in detail. We prove the structure theorem for these algebras and classify low-dimensional non-trivial generalized Hessian pre-Lie algebras over \(\mathbb C\).

Throughout this paper, all vector spaces are assumed to be finite-dimensional
over an algebraically closed field \(F\) of characteristic zero. For any vector space \(V\), the natural pairing between \(V^*\) and \(V\) is denoted by \(\langle f,x\rangle=f(x)\) for all \(f\in V^*\) and \(x\in V\).

\section{($\omega$-)pre-Lie algebra structures on the dual spaces of pre-Lie algebras}\label{22}
In this section, we recall how solutions of the classical Yang--Baxter
equation of a pre-Lie algebra \(A\) induce pre-Lie algebra structures on the dual space \(A^*\). This motivates the introduction of the \(A\)-generalized Yang--Baxter equation below. We then show that solutions of this generalized equation give rise to multiplicative (\(\omega\)-)pre-Lie algebra structures on \(A^*\). 

\begin{defn}\cite{Vinberg1963}
A vector space \(A\) endowed with a bilinear product \(\cdot\) is called a
pre-Lie algebra if
\begin{equation*}
        (x\cdot y)\cdot z-x\cdot (y\cdot z)
        =
        (y\cdot x)\cdot z-y\cdot (x\cdot z)
\end{equation*}
for all \(x,y,z\in A\). Equivalently, the associator \((x,y,z):=(x\cdot y)\cdot z-x\cdot (y\cdot z)\) is symmetric in the first two variables.
\end{defn}

\begin{defn}\cite{ChenWu2023}
Let \(A\) be a vector space endowed with a bilinear product \(\cdot\), and let \(\omega\) be a bilinear form on \(A\). The triple \((A,\cdot,\omega)\) is called an \(\omega\)-pre-Lie algebra if
\begin{equation*}
        (x\cdot y)\cdot z-x\cdot (y\cdot z)
        -(y\cdot x)\cdot z+y\cdot (x\cdot z)
        = \omega(x,y)z
\end{equation*}
for all \(x,y,z\in A\).
\end{defn}

\begin{rem}
Let \((A,\cdot,\omega)\) be an \(\omega\)-pre-Lie algebra and define
\([x,y]=x\cdot y-y\cdot x\) for all \(x,y\in A\). Then \((A,[\cdot,\cdot],\omega)\) is an \(\omega\)-Lie algebra \cite{Nurowski2007}, called the sub-adjacent \(\omega\)-Lie algebra of \((A,\cdot,\omega)\). It is denoted by \(A_{\mathrm c}\). 
\end{rem}

\begin{defn}
Let \((A,\cdot,\omega)\) be an \(\omega\)-pre-Lie algebra. It is called
multiplicative if there exists \(\theta\in A^*\) such that
\begin{equation*}
        \omega(x,y)=\theta(x\cdot y)-\theta(y\cdot x)
\end{equation*}
for all \(x,y\in A\). In this case, we called \((A,\cdot,\theta)\) is a multiplicative
\(\omega\)-pre-Lie algebra.
\end{defn}

\begin{rem}
Let \((A,\cdot,\theta)\) be a multiplicative \(\omega\)-left-symmetric
algebra, the sub-adjacent \(\omega\)-Lie algebra \(A_{\mathrm c}=(A,[\cdot,\cdot],\theta)\)
is a multiplicative \(\omega\)-Lie algebra \cite{Zusmanovich2010}.
\end{rem}

\begin{defn}\cite{Bai2008}
Let \((A,\cdot)\) be a pre-Lie algebra. For \(r=\sum_i x_i\otimes y_i\in A\otimes A\), define \([[r,r]]\in A^{\otimes 3}\) by
\begin{equation*}
\begin{aligned}
[[r,r]] = \sum_{i,j}\Big(
& x_i\cdot x_j\otimes y_j\otimes y_i
 - y_j\otimes x_i\cdot x_j\otimes y_i
 + x_j\otimes [x_i,y_j]\otimes y_i  \\
& - [x_i,y_j]\otimes x_j\otimes y_i
 - x_i\otimes x_j\otimes [y_i,y_j]
\Big),
\end{aligned}
\end{equation*}
where the brackets are taken in \(A_{\mathrm c}\). The equation \([[r,r]]=0\) is called the classical Yang--Baxter equation of the pre-Lie algebra \(A\).
\end{defn}

\begin{rem}
The classical Yang--Baxter equation of a pre-Lie algebra \(A\) is equivalent to the equation 
\begin{equation*}
[[r,r]] = r_{13} \cdot r_{12} - r_{23} \cdot r_{21} + [r_{23}, r_{12}] - [r_{13}, r_{21}] - [r_{13}, r_{23}] = 0
\end{equation*}
in the universal enveloping algebra $U(A_c)$, where
\begin{align*}
    r_{12} &= \sum_i x_i \otimes y_i \otimes 1, \quad r_{21} = \sum_i y_i \otimes x_i \otimes 1, \\
    r_{13} &= \sum_i x_i \otimes 1 \otimes y_i, \quad r_{23} = \sum_i 1 \otimes x_i \otimes y_i.
\end{align*}
\end{rem}

\begin{thm}\label{2.6}\cite{Bai2008}
Let \((A,\cdot)\) be a pre-Lie algebra and let \(r=\sum_i x_i\otimes y_i\in A\otimes A\). Define \(\alpha_r:A\to A\otimes A\) by
\begin{equation*}
        \alpha_r(x)
        =
        (L_x\otimes \operatorname{id}
        +\operatorname{id}\otimes \operatorname{ad}_x)r
        \quad \forall x\in A,
\end{equation*}
where \(L_x\) denotes left multiplication and \(\operatorname{ad}_x\) is taken in \(A_{\mathrm c}\). Define a product \(\circ\) on \(A^*\) by
\begin{equation*}
        \langle a\circ b,x\rangle
        =
        \langle a\otimes b,\alpha_r(x)\rangle
        \quad \forall a,b\in A^*,\ x\in A .
\end{equation*}
Then \((A^*,\circ)\) is a pre-Lie algebra if and only if
\begin{equation*}
\begin{aligned}
        Q(x)[[r,r]] + \sum_i \bigl(P(x\cdot x_i)-P(x)P(x_i)\bigr)
        (r-\sigma(r))\otimes y_i = 0
\end{aligned}
\end{equation*}
for all \(x\in A\), where \(\sigma:A\otimes A\to A\otimes A\) is the flip map,
\(Q(x)=L_x\otimes \operatorname{id}\otimes \operatorname{id}
+\operatorname{id}\otimes L_x\otimes \operatorname{id}
+\operatorname{id}\otimes \operatorname{id}\otimes \operatorname{ad}_x\), and
\(P(x)=L_x\otimes \operatorname{id}+\operatorname{id}\otimes L_x\).
\end{thm}

\begin{rem}
If \(r=\sigma(r)\), then \([[r,r]]=0\) is equivalent to
\begin{equation*}
        -r_{12}\cdot r_{13}
        +r_{12}\cdot r_{23}
        +[r_{13},r_{23}]
        =0 .
\end{equation*}
This equation is called the classical \(S\)-equation of the pre-Lie algebra \(A\) \cite{Bai2008}.
\end{rem}

Symmetric solutions of the classical \(S\)-equation provide an immediate class for which the condition in Theorem~\ref{2.6} is automatically satisfied.

\begin{cor}\label{cor:symmetric-S-dual-LSA}
Let \((A,\cdot)\) be a pre-Lie algebra and let \(r=\sum_i x_i\otimes y_i\in A\otimes A\) be a symmetric solution of the classical \(S\)-equation. Let \(\alpha_r:A\to A\otimes A\) and the product \(\circ\) on \(A^*\) be defined as in Theorem~\ref{2.6}. Then \((A^*,\circ)\) is a pre-Lie algebra.
\end{cor}

Let \((A,\cdot)\) be a pre-Lie algebra and let \(r=\sum_i x_i\otimes y_i\in A\otimes A\). Define linear maps \(r_+,r_-:A^*\to A\) by
\[
        r_+(\xi)=\sum_i\langle \xi,x_i\rangle y_i,
        \qquad
        r_-(\xi)=\sum_i\langle \xi,y_i\rangle x_i .
\]
for all \(\xi\in A^*\). Then, for all \(\xi,\eta\in A^*\),
\[
\begin{aligned}
        \bigl\langle \xi,(r_+-r_-)(\eta)\bigr\rangle
        &=
        \sum_i \bigl(\langle \eta,x_i\rangle\langle \xi,y_i\rangle
        - \langle \eta,y_i\rangle\langle \xi,x_i\rangle
        \bigr)  \\
        &= \langle \eta\otimes \xi, r-\sigma(r)\rangle .
\end{aligned}
\]
Consequently, \(r\) is symmetric if and only if \(r_{+}-r_{-}=0\). The extremal case of non-symmetry occurs when \(r_{+}-r_{-}\) is a linear isomorphism.
\begin{defn}\cite{WangBaiLiuSheng2024}
Let \((A,\cdot)\) be a pre-Lie algebra. An element \(s\in A\otimes A\)
is called \((L,\operatorname{ad})\)-invariant if
\begin{equation*}
        (L_x\otimes \operatorname{id}
        +\operatorname{id}\otimes \operatorname{ad}_x)s=0
\end{equation*}
for all \(x\in A\). A solution \(r\in A\otimes A\) of the classical
Yang--Baxter equation of \(A\) is called factorizable if \(r-\sigma(r)\) is
\((L,\operatorname{ad})\)-invariant and \(r_+-r_-\) is a linear isomorphism.
\end{defn}

Factorizable solutions provide another natural class for which the condition in Theorem~\ref{2.6} is automatically satisfied. The key point is the following
lemma.

\begin{lem}\cite{WangBaiLiuSheng2024}\label{lem:invariant-skew-part}
Let \((A,\cdot)\) be a pre-Lie algebra and let \(r\in A\otimes A\).
If \(r-\sigma(r)\) is \((L,\operatorname{ad})\)-invariant, then
\begin{equation*}
        \bigl(P(x\cdot y)-P(x)P(y)\bigr)(r-\sigma(r))=0
\end{equation*}
for all \(x,y\in A\), where \(P\) is defined in Theorem~\ref{2.6}.
\end{lem}

Therefore, we get the following corollary.

\begin{cor}\label{cor:factorizable-dual-LSA}
Let \((A,\cdot)\) be a pre-Lie algebra and let \(r=\sum_i x_i\otimes y_i\in A\otimes A\) be a factorizable solution of the classical Yang--Baxter equation of \(A\). Let \(\alpha_r:A\to A\otimes A\) and the product \(\circ\) on \(A^*\) be defined as in Theorem~\ref{2.6}. Then \((A^*,\circ)\) is a pre-Lie algebra.
\end{cor}

\begin{defn}
For a pre-Lie algebra \((A,\cdot)\), define its right annihilator by
\begin{equation*}
        \operatorname{Ann}_{\mathrm R}(A)
        = \{a\in A\mid x\cdot a=0\ \forall x\in A\}
\end{equation*}
and its annihilator by
\begin{equation*}
        \operatorname{Ann}(A)
        =
        \{a\in A\mid a\cdot x=x\cdot a=0\ \forall x\in A\}.
\end{equation*}
\end{defn}

We next introduce the \(A\)-generalized Yang--Baxter equation of pre-Lie algebras, which extends the classical case.

\begin{defn}
Let \((A,\cdot)\) be a pre-Lie algebra, let \(u\in \operatorname{Ann}_{\mathrm R}(A)\), and let \(r=\sum_i x_i\otimes y_i\in A\otimes A\). Define \([[r,r]]_u\in A^{\otimes 3}\) by
\begin{equation}\label{eq:AYBE}
\begin{aligned}
[[r,r]]_u
=
\sum_{i,j}\Big(
& x_i\cdot x_j\otimes y_j\otimes y_i
 - y_j\otimes x_i\cdot x_j\otimes y_i
 + x_j\otimes [x_i,y_j]\otimes y_i  \\
& - [x_i,y_j]\otimes x_j\otimes y_i
 - x_i\otimes x_j\otimes [y_i,y_j]
\Big) + \sum_i \bigl( u\otimes x_i\otimes y_i
- x_i\otimes u\otimes y_i \bigr),
\end{aligned}
\end{equation}
where the brackets are taken in \(A_{\mathrm c}\). The equation \([[r,r]]_u=0\) is called the \emph{\(A\)-generalized Yang--Baxter equation} of the pre-Lie algebra \(A\). When \(u\in\operatorname{Ann}_{\mathrm R}(A)\) is fixed, we also call it the \emph{\(u\)-generalized Yang--Baxter equation}. In particular, if \(u\in\operatorname{Ann}(A)\), we use the same terminology.
\end{defn}

\begin{rem}
For \(u=0\), \eqref{eq:AYBE} reduces to the classical Yang--Baxter equation of the pre-Lie algebra \(A\). In this paper, we are mainly interested in the nontrivial case \(u\neq 0\). With the tensor notation introduced above, \([[r,r]]_u\) can be written as
\begin{equation*}
\begin{aligned}
        [[r,r]]_u
        = &\, r_{13}\cdot r_{12}
        -r_{23}\cdot r_{21}
        +[r_{23},r_{12}]
        -[r_{13},r_{21}]
        -[r_{13},r_{23}] + (u\otimes 1\otimes 1)r_{23}
        -(1\otimes u\otimes 1)r_{13}.
\end{aligned}
\end{equation*}
\end{rem}

For a linear map \(\partial:A\to A\otimes A\) and an element \(u\in A\), define \(J_\partial^u:A\to A^{\otimes 3}\) as follows: for all \(x\in A\),
\begin{equation*}
\begin{aligned}
J_\partial^u(x)
=
&(\partial\otimes \operatorname{id})\partial(x)
-(\operatorname{id}\otimes \partial)\partial(x)
+u\otimes \partial(x)                                      \\
&-(\sigma\otimes \operatorname{id})
\bigl((\partial\otimes \operatorname{id})\partial(x)
-(\operatorname{id}\otimes \partial)\partial(x)
+u\otimes \partial(x) \bigr).
\end{aligned}
\end{equation*}

\begin{lem}\label{lem:u-Jacobi-alpha}
Let \((A,\cdot)\) be a pre-Lie algebra, let \(u\in \operatorname{Ann}_{\mathrm R}(A)\), and let \(r=\sum_i x_i\otimes y_i\in A\otimes A\). For the map \(\alpha_r:A\to A\otimes A\) defined in Theorem~\ref{2.6}, we have
\begin{equation*}
\begin{aligned}
J_{\alpha_r}^u(x)
= Q(x)[[r,r]]_u
+ \sum_i \bigl( (P(x\cdot x_i)-P(x)P(x_i))(r-\sigma(r)) \bigr)\otimes y_i
\end{aligned}
\end{equation*}
for all \(x\in A\), where \(P\) and \(Q\) are defined in Theorem~\ref{2.6}.
\end{lem}

\begin{proof}
Set
\begin{equation*}
        T_u = \sum_i \bigl(u\otimes x_i\otimes y_i - x_i\otimes u\otimes y_i\bigr),
\end{equation*}
so that \([[r,r]]_u=[[r,r]]+T_u\). By the definition of \(J_{\alpha_r}^u\), we have
\begin{equation*}
        J_{\alpha_r}^u(x)
        = J_{\alpha_r}^0(x)
        + u\otimes \alpha_r(x)
        -(\sigma\otimes \operatorname{id})
        \bigl(u\otimes \alpha_r(x)\bigr)
\end{equation*}
for all \(x\in A\). Since $\alpha_r(x) = \sum_i \bigl(
        (x\cdot x_i)\otimes y_i
        + x_i\otimes [x,y_i]\bigr),$
we get
\begin{equation*}
\begin{aligned}
&u\otimes \alpha_r(x)
-
(\sigma\otimes \operatorname{id})
\bigl(u\otimes \alpha_r(x)\bigr)                         \\
={}&
\sum_i
\Bigl(u\otimes (x\cdot x_i)\otimes y_i
+ u\otimes x_i\otimes [x,y_i]
- (x\cdot x_i)\otimes u\otimes y_i
- x_i\otimes u\otimes [x,y_i]
\Bigr).
\end{aligned}
\end{equation*}
On the other hand, \(u\in \operatorname{Ann}_{\mathrm R}(A)\) implies
\(x\cdot u=0\) for all \(x\in A\). Hence applying \(Q(x)\) to \(T_u\) gives
the same expression. Thus
\begin{equation*}
        J_{\alpha_r}^u(x)=J_{\alpha_r}^0(x)+Q(x)T_u .
\end{equation*}
Using the classical identity corresponding to the case \(u=0\),
\begin{equation*}
        J_{\alpha_r}^0(x)
        = Q(x)[[r,r]] +   \sum_i\bigl( (P(x\cdot x_i)-P(x)P(x_i))(r-\sigma(r))
        \bigr)\otimes y_i ,
\end{equation*}
and using \([[r,r]]_u=[[r,r]]+T_u\), we obtain the desired formula.
\end{proof}

\begin{thm}\label{2.12}
Let \((A,\cdot)\) be a pre-Lie algebra, let
\(u\in \operatorname{Ann}_{\mathrm R}(A)\), and let \(r=\sum_i x_i\otimes y_i\in A\otimes A\). Let \(\alpha_r:A\to A\otimes A\) be the map defined in Theorem~\ref{2.6}. Define a product \(\circ\) on \(A^*\) by
\begin{equation*}
        \langle a\circ b,x\rangle
        = \langle a\otimes b,\alpha_r(x)\rangle + \langle a,u\rangle\langle b,x\rangle
\end{equation*}
for all \(a,b\in A^*\) and \(x\in A\), and define a linear functional \(\theta^*:A^*\to F\) by
\begin{equation*}
        \theta^*(a)=\langle a,u\rangle
\end{equation*}
for all \(a\in A^*\). Then \((A^*,\circ,\theta^*)\) is a multiplicative \(\omega\)-pre-Lie algebra if and only if
\begin{equation*}
\begin{aligned}
        Q(x)[[r,r]]_u
        + \sum_i \bigl(P(x\cdot x_i)-P(x)P(x_i)\bigr)(r-\sigma(r))\otimes y_i
        = 0
\end{aligned}
\end{equation*}
for all \(x\in A\), where \(P\) and \(Q\) are defined in Theorem~\ref{2.6}.
\end{thm}

\begin{proof}
By the definition of a multiplicative \(\omega\)-pre-Lie algebra, it suffices to verify the identity
\begin{equation*}
\begin{aligned}
&(a\circ b)\circ c-a\circ(b\circ c)
-(b\circ a)\circ c+b\circ(a\circ c)-\bigl(\theta^*(a\circ b)-\theta^*(b\circ a)\bigr)c=0
\end{aligned}
\end{equation*}
for all \(a,b,c\in A^*\). Pairing the left-hand side with \(x\in A\), a direct expansion from the definition of \(\circ\) gives
\begin{equation*}
\begin{aligned}
&\Big\langle
(a\circ b)\circ c-a\circ(b\circ c)
-(b\circ a)\circ c+b\circ(a\circ c) -\bigl(\theta^*(a\circ b)-\theta^*(b\circ a)\bigr)c,\ x
\Big\rangle                                      \\
={}& \big\langle a\otimes b\otimes c,\,
J_{\alpha_r}^u(x) \big\rangle .
\end{aligned}
\end{equation*}
Indeed, the terms independent of \(u\) give the classical expression corresponding to \(J_{\alpha_r}^0(x)\), while the additional term
\(\langle a,u\rangle b\) in the product contributes exactly
\begin{equation*}
        u\otimes \alpha_r(x)
        - (\sigma\otimes \operatorname{id})
        \bigl(u\otimes \alpha_r(x)\bigr).
\end{equation*}
Thus \((A^*,\circ,\theta^*)\) is a multiplicative \(\omega\)-pre-Lie algebra if and only if \(J_{\alpha_r}^u(x)=0\) for every \(x\in A\). By Lemma~\ref{lem:u-Jacobi-alpha}, this is equivalent to the stated tensor identity.
\end{proof}

\begin{rem}
If \(r=\sigma(r)\), then the \emph{\(u\)-generalized Yang--Baxter equation} \([[r,r]]_u=0\) is equivalently written as
\begin{equation*}
        -r_{12}\cdot r_{13}
        +r_{12}\cdot r_{23}
        +[r_{13},r_{23}]
        -(u\otimes 1\otimes 1)r_{23}
        +(1\otimes u\otimes 1)r_{13}
        = 0.
\end{equation*}
We call this equation the \emph{\(u\)-generalized \(S\)-equation} of the
pre-Lie algebra \(A\).
\end{rem}

Symmetric solutions of the \(u\)-generalized \(S\)-equation provide an
immediate class for which the condition in Theorem~\ref{2.12} is automatically satisfied.

\begin{cor}\label{cor:u-S-dual-omega-LSA}
Let \((A,\cdot)\) be a pre-Lie algebra, let \(u\in \operatorname{Ann}_{\mathrm R}(A)\), and let \(r=\sum_i x_i\otimes y_i\in A\otimes A\) be a symmetric solution of the \(u\)-generalized \(S\)-equation. Let the product \(\circ\) on \(A^*\) and the linear functional \(\theta^*:A^*\to F\) be defined as in Theorem~\ref{2.12}. Then \((A^*,\circ,\theta^*)\) is a multiplicative \(\omega\)-pre-Lie algebra.
\end{cor}

A solution \(r\) of the \(u\)-generalized Yang--Baxter equation is called
factorizable if \(r_+-r_-:A^*\to A\) is a linear isomorphism and \(r-\sigma(r)\) is \((L,\operatorname{ad})\)-invariant. By Lemma~\ref{lem:invariant-skew-part}, factorizable solutions provide another
class for which the condition in Theorem~\ref{2.12} is automatically satisfied.

\begin{cor}\label{cor:u-factorizable-dual-omega-LSA}
Let \((A,\cdot)\) be a pre-Lie algebra, let \(u\in \operatorname{Ann}_{\mathrm R}(A)\), and let \(r=\sum_i x_i\otimes y_i\in A\otimes A\) be a factorizable solution of the \(u\)-generalized Yang--Baxter equation. Let the product \(\circ\) on \(A^*\) and the linear functional \(\theta^*:A^*\to F\) be defined as in Theorem~\ref{2.12}. Then \((A^*,\circ,\theta^*)\) is a multiplicative \(\omega\)-pre-Lie algebra.
\end{cor}

When \(u\in\operatorname{Ann}(A)\), the \(\omega\)-term in Theorem~\ref{2.12} vanishes. Hence the multiplicative \(\omega\)-pre-Lie condition reduces to the ordinary left-symmetric identity.

\begin{thm}\label{2.18}
Let \((A,\cdot)\) be a pre-Lie algebra, let \(u\in \operatorname{Ann}(A)\), and let \(r=\sum_i x_i\otimes y_i\in A\otimes A\). Let the product \(\circ\) on \(A^*\) be defined as in Theorem~\ref{2.12}. Then \((A^*,\circ)\) is a
pre-Lie algebra if and only if
\begin{equation*}
\begin{aligned}
        Q(x)[[r,r]]_u +\sum_i \bigl(P(x\cdot x_i)-P(x)P(x_i)\bigr)(r-\sigma(r))\otimes y_i
        = 0
\end{aligned}
\end{equation*}
for all \(x\in A\), where \(P\) and \(Q\) are defined in Theorem~\ref{2.6}.
\end{thm}

\begin{proof}
Let \(\theta^*\) be the linear functional defined in Theorem~\ref{2.12}. Since \(u\in\operatorname{Ann}(A)\), we have \(L_u=0\) and \(\operatorname{ad}_u=0\). Hence \(\alpha_r(u)=0\). Therefore
\begin{equation*}
        \theta^*(a\circ b)
        = \langle a\circ b,u\rangle
        = \langle a,u\rangle\langle b,u\rangle
        = \theta^*(b\circ a)
\end{equation*}
for all \(a,b\in A^*\). Thus the induced \(\omega\)-term vanishes, and the
multiplicative \(\omega\)-pre-Lie identity in Theorem~\ref{2.12}
reduces exactly to the ordinary pre-Lie identity. The result follows
from Theorem~\ref{2.12}.
\end{proof}

Symmetric solutions of the \(u\)-generalized \(S\)-equation provide an
immediate class for which the condition in Theorem~\ref{2.18} is automatically satisfied.

\begin{cor}\label{cor:annihilator-u-S-dual-LSA}
Let \((A,\cdot)\) be a pre-Lie algebra, let \(u\in \operatorname{Ann}(A)\), and let \(r=\sum_i x_i\otimes y_i\in A\otimes A\) be a symmetric solution of the \(u\)-generalized \(S\)-equation. Let the product \(\circ\) on \(A^*\)
be defined as in Theorem~\ref{2.12}. Then \((A^*,\circ)\) is a pre-Lie algebra.
\end{cor}

By Lemma~\ref{lem:invariant-skew-part}, factorizable solutions of the
\(u\)-generalized Yang--Baxter equation provide another class for which
the condition in Theorem~\ref{2.18} is automatically satisfied.

\begin{cor}\label{cor:annihilator-u-factorizable-dual-LSA}
Let \((A,\cdot)\) be a pre-Lie algebra, let \(u\in \operatorname{Ann}(A)\), and let \(r=\sum_i x_i\otimes y_i\in A\otimes A\) be a factorizable solution of the \(u\)-generalized Yang--Baxter equation. Let the product \(\circ\) on
\(A^*\) be defined as in Theorem~\ref{2.12}. Then \((A^*,\circ)\) is a
pre-Lie algebra.
\end{cor}

\section{Solutions of the $A$-generalized Yang-Baxter equations}\label{33}

In this section, we study symmetric and factorizable solutions of the $A$-generalized Yang--Baxter equation. We start with the operator form of symmetric solutions on arbitrary pre-Lie algebras.

\subsection{Symmetric Solutions on Arbitrary pre-Lie algebras}

For arbitrary pre-Lie algebras, symmetric solutions of the classical
\(S\)-equation can be characterized by relative Rota--Baxter operators
associated with the coregular representation.

\begin{defn}\cite{Bai2008}
A representation of a pre-Lie algebra \((A,\cdot)\) is a triple \((\rho,\mu,V)\), where \(V\) is a vector space, \(\rho:A\to\mathfrak{gl}(V)\) is a representation of the sub-adjacent Lie algebra \(A_{\mathrm c}\), and \(\mu:A\to\mathfrak{gl}(V)\) is a linear map satisfying
\begin{equation*}
    \rho(x)\mu(y)-\mu(y)\rho(x)
    =
    \mu(x\cdot y)-\mu(y)\mu(x)
\end{equation*}
for all \(x,y\in A\).
\end{defn}

\begin{rem}
Let \((A,\cdot)\) be a pre-Lie algebra. The regular representation of \(A\) is \((L,R,A)\), where \(L,R:A\to\mathfrak{gl}(A)\) are the left and right multiplication operators, namely \(L_x y=x\cdot y\) and \(R_x y=y\cdot x\) for all \(x,y\in A\).
\end{rem}

\begin{defn}\cite{Bai2008}
Let \((\rho,\mu,V)\) be a representation of a pre-Lie algebra \((A,\cdot)\). The dual representation on \(V^*\) is \((\rho^*-\mu^*,-\mu^*,V^*)\), where \(\rho^*,\mu^*:A\to\mathfrak{gl}(V^*)\) are defined by
\begin{equation*}
\begin{aligned}
        \langle \rho^*(x)\xi,v\rangle
        &= -\langle \xi,\rho(x)v\rangle, \\
        \langle \mu^*(x)\xi,v\rangle
        &= -\langle \xi,\mu(x)v\rangle
\end{aligned}
\end{equation*}
for all \(x\in A\), \(\xi\in V^*\), and \(v\in V\).
\end{defn}

\begin{rem}
Let \((A,\cdot)\) be a pre-Lie algebra. The coregular representation of \(A\) is \((\operatorname{ad}^*,-R^*,A^*)\), where \(\operatorname{ad}_x^*=L_x^*-R_x^*\) is the coadjoint representation of \(A_{\mathrm c}\).
\end{rem}

\begin{thm}\cite{Liu}
Let \((A,\cdot)\) be a pre-Lie algebra, and let \(T:A^*\to A\) be a linear map satisfying \(\langle \xi,T(\eta)\rangle=\langle \eta,T(\xi)\rangle\)
for all \(\xi,\eta\in A^*\). Let \(\{e_i\}\) be a basis of \(A\), and let
\(\{e^i\}\) be its dual basis. Then
\begin{equation*}
        r=\sum_i T(e^i)\otimes e_i
\end{equation*}
is a symmetric solution of the classical \(S\)-equation if and only if \(T\)
is a relative Rota--Baxter operator associated to the coregular representation, namely
\begin{equation*}
        T(\xi)\cdot T(\eta)
        = T\bigl( \operatorname{ad}_{T(\xi)}^*\eta
        - R_{T(\eta)}^*\xi
        \bigr)
\end{equation*}
for all \(\xi,\eta\in A^*\).
\end{thm}

We generalize the above result to obtain symmetric solutions of the \(u\)-generalized \(S\)-equation.

\begin{thm}\label{thm:generalized-O-operator-S-equation}
Let \((A,\cdot)\) be a pre-Lie algebra, let \(0\neq u\in\operatorname{Ann}(A)\), and let \(T:A^*\to A\) be a linear map satisfying \(\langle \xi,T(\eta)\rangle=\langle \eta,T(\xi)\rangle\) for all \(\xi,\eta\in A^*\). Let \(\{e_i\}\) be a basis of \(A\), and let \(\{e^i\}\) be the dual basis of \(A^*\). Then
\begin{equation*}
        r=\sum_i T(e^i)\otimes e_i
\end{equation*}
is a symmetric solution of the \(u\)-generalized \(S\)-equation if and
only if
\begin{equation}\label{eq:generalized-O-operator-dual}
\begin{aligned}
T(\xi)\cdot T(\eta)
={}&
T\bigl( \operatorname{ad}_{T(\xi)}^*\eta - R_{T(\eta)}^*\xi \bigr) + \langle \xi,u\rangle T(\eta) - \langle \xi,T(\eta)\rangle u
\end{aligned}
\end{equation}
for all \(\xi,\eta\in A^*\).
\end{thm}

\begin{proof}
The symmetry of \(T\) is equivalent to the symmetry of $r=\sum_iT(e^i)\otimes e_i$. Pairing
\begin{equation*}
        -r_{12}\cdot r_{13}
        +r_{12}\cdot r_{23}
        +[r_{13},r_{23}]
        -(u\otimes1\otimes1)r_{23}
        +(1\otimes u\otimes1)r_{13}
\end{equation*}
with \(\xi\otimes\eta\otimes\zeta\), we obtain
\begin{equation*}
\begin{aligned}
-&\Big\langle \xi,\,
T(\eta)\cdot T(\zeta)
- T\bigl(\operatorname{ad}_{T(\eta)}^*\zeta
- R_{T(\zeta)}^*\eta\bigr)
+\langle \zeta,T(\eta)\rangle u
-\langle \eta,u\rangle T(\zeta)
\Big\rangle .
\end{aligned}
\end{equation*}
Since the pairing between \(A^*\) and \(A\) is nondegenerate, the \(u\)-generalized \(S\)-equation is equivalent to
\begin{equation*}
\begin{aligned}
T(\eta)\cdot T(\zeta)
={}& T\bigl(\operatorname{ad}_{T(\eta)}^*\zeta
- R_{T(\zeta)}^*\eta
\bigr) + \langle \eta,u\rangle T(\zeta)
- \langle \zeta,T(\eta)\rangle u .
\end{aligned}
\end{equation*}
Renaming \(\eta,\zeta\) as \(\xi,\eta\), and using \(\langle \eta,T(\xi)\rangle=\langle \xi,T(\eta)\rangle\) gives \eqref{eq:generalized-O-operator-dual}.
\end{proof}

\begin{defn}
Let \((A,\cdot)\) be a pre-Lie algebra. A linear map \(T:A^*\to A\) satisfying \eqref{eq:generalized-O-operator-dual} for all \(\xi,\eta\in A^*\) is called a \emph{\(u\)-generalized relative Rota--Baxter operators associated with the coregular representation}.
\end{defn}

\subsection{\(u\)-Generalized Hessian pre-Lie algebra and Symmetric Solutions}

In the same way that Hessian pre-Lie algebras give rise to relative Rota--Baxter operators, \(u\)-generalized Hessian pre-Lie algebras give rise to \(u\)-generalized relative Rota--Baxter operators.

\begin{defn}\label{def:generalized-Hessian-LSA}
Let \((A,\cdot)\) be a pre-Lie algebra and let \(u\in\operatorname{Ann}(A)\). A symmetric bilinear form \(\gamma\) on \(A\) is called a \emph{\(u\)-generalized \(2\)-cocycle} if
\begin{equation*}
\begin{aligned}
&\gamma(x\cdot y,z)-\gamma(x,y\cdot z)
-\gamma(y\cdot x,z)+\gamma(y,x\cdot z) \\
&\quad
-\gamma(x,u)\gamma(y,z)+\gamma(y,u)\gamma(x,z)=0
\end{aligned}
\end{equation*}
for all \(x,y,z\in A\). A \emph{\(u\)-generalized Hessian pre-Lie algebra} is a quadruple \((A,\cdot,\gamma,u)\), where \((A,\cdot)\) is a pre-Lie algebra and \(\gamma\) is a nondegenerate \emph{\(u\)-generalized \(2\)-cocycle}.
\end{defn}

\begin{rem}
The notion of a \emph{\(u\)-generalize Hessian pre-Lie algebra} has the following geometric interpretation. Let \((A,\cdot)\) be realized as the algebra of left-invariant vector fields on the simply connected Lie group with Lie algebra \(A_{\mathrm c}\), whose left-invariant flat torsion-free connection is determined by \(\nabla_x y=x\cdot y\). A nondegenerate symmetric bilinear form \(\gamma\) corresponds to a left-invariant pseudo-Riemannian metric \(g\). For left-invariant vector fields, we have
\begin{equation*}
        (\nabla_x\gamma)(y,z)
        =
        -\gamma(x\cdot y,z)-\gamma(y,x\cdot z).
\end{equation*}
Define \(\theta\in A^*\) by \(\theta(x)=\gamma(x,u)\). Then the \(u\)-generalized \(2\)-cocycle condition is equivalent to the total symmetry of
\begin{equation*}
        \nabla g+\theta\otimes g .
\end{equation*}
Indeed, the difference between the values of this tensor on \((x,y,z)\) and
\((y,x,z)\) is the negative of the left-hand side of the \(u\)-generalized \(2\)-cocycle identity. Moreover, taking \(z=u\) in that identity and using \(u\in\operatorname{Ann}(A)\), we get
\begin{equation*}
        \gamma(x\cdot y,u)-\gamma(y\cdot x,u)=0,
\end{equation*}
equivalently, \(\theta([x,y])=0\). Hence \(\theta\) is closed as a left-invariant \(1\)-form. If \(\vartheta=-\theta\), then
\begin{equation*}
        \nabla g-\vartheta\otimes g
\end{equation*}
is totally symmetric and \(\vartheta\) is closed. Thus, in the positive
definite case, \(u\)-generalized Hessian pre-Lie algebras correspond to the infinitesimal form of locally conformally Hessian manifolds in the sense of Osipov~\cite{Osipov2023}.
\end{rem}

\begin{prop}\label{prop:generalized-Hessian-O-operator}
Let \((A,\cdot,\gamma,u)\) be a \(u\)-generalized Hessian pre-Lie algebra. Define \(\gamma^\sharp:A\to A^*\) by
\[
\langle \gamma^\sharp(x),y\rangle=\gamma(x,y)
\]
for all \(x,y\in A\). Then \(T_\gamma=(\gamma^\sharp)^{-1}:A^*\to A\) is a
\(u\)-generalized relative Rota--Baxter operators associated with the coregular representation. 
\end{prop}
\begin{proof}
The nondegeneracy of \(\gamma\) implies that \(\gamma^\sharp\) is an
isomorphism. Let \(a=T_\gamma(\xi)\) and \(b=T_\gamma(\eta)\). Then
\(\langle \xi,y\rangle=\gamma(a,y)\) and \(\langle \eta,y\rangle=\gamma(b,y)\) for all \(y\in A\). Since \(\gamma\) is symmetric, we have
\begin{equation*}
        \langle \xi,T_\gamma(\eta)\rangle
        =
        \gamma(a,b)
        =
        \gamma(b,a)
        =
        \langle \eta,T_\gamma(\xi)\rangle .
\end{equation*}
For any \(z\in A\), using the definitions of \(\operatorname{ad}^*\) and
\(R^*\), we get
\begin{equation*}
\begin{aligned}
&\gamma\Big(
T_\gamma\bigl(\operatorname{ad}_a^*\eta-R_b^*\xi\bigr)
+\langle \xi,u\rangle b
-\langle \xi,b\rangle u,\ z
\Big)  \\
={}& -\gamma(b,[a,z]) +\gamma(a,z\cdot b) +\gamma(a,u)\gamma(b,z)
-\gamma(a,b)\gamma(u,z).
\end{aligned}
\end{equation*}
By the \(u\)-generalized \(2\)-cocycle identity applied to \((a,z,b)\), the right-hand side is equal to \(\gamma(a\cdot b,z)\). Hence, by the nondegeneracy of \(\gamma\),
\begin{equation*}
\begin{aligned}
T_\gamma(\xi)\cdot T_\gamma(\eta)
={}&
T_\gamma\bigl(
\operatorname{ad}_{T_\gamma(\xi)}^*\eta
-
R_{T_\gamma(\eta)}^*\xi
\bigr)
+
\langle \xi,u\rangle T_\gamma(\eta)
-
\langle \xi,T_\gamma(\eta)\rangle u .
\end{aligned}
\end{equation*}
Thus \(T_\gamma\) satisfies \eqref{eq:generalized-O-operator-dual}.
\end{proof}

\begin{cor}\label{cor:GHLSA-symmetric-solution}
Let \((A,\cdot,\gamma,u)\) be a \(u\)-generalized Hessian pre-Lie algebra with \(u\neq 0\). Let \(\{e_i\}\) be a basis of \(A\), and let \(\{e^i\}\) be its dual basis. Then
\[
r_\gamma=\sum_i T_\gamma(e^i)\otimes e_i
\]
is a symmetric solution of the \(u\)-generalized \(S\)-equation.
\end{cor}

Conversely, a symmetric solution of the \(u\)-generalized \(S\)-equation also gives rise to a \(u\)-generalized Hessian pre-Lie algebra.

\begin{prop}\label{thm:nondegenerate-S-solution-generalized-Hessian}
Let \((A,\cdot)\) be a pre-Lie algebra, let \(0\neq u\in\operatorname{Ann}(A)\), and let \(r=\sum_i x_i\otimes y_i\in A\otimes A\) be a symmetric solution of the \(u\)-generalized \(S\)-equation. Define \(r^\sharp:A^*\to A\) by
\begin{equation*}
        \langle \xi,r^\sharp(\eta)\rangle
        =
        \langle \xi\otimes\eta,r\rangle
\end{equation*}
for all \(\xi,\eta\in A^*\). If \(r^\sharp\) is a linear isomorphism, then
\begin{equation*}
        \gamma(x,y)
        =
        \bigl\langle (r^\sharp)^{-1}(x),y\bigr\rangle
\end{equation*}
for all \(x,y\in A\) defines a \(u\)-generalized Hessian pre-Lie algebra \((A,\cdot,\gamma,u)\).
\end{prop}

\begin{proof}
The nondegeneracy of \(\gamma\) follows from the invertibility of
\(r^\sharp\). Since \(r\) is symmetric, \(r^\sharp\) is symmetric with respect to the natural pairing. Hence, for \(x=r^\sharp(\xi)\) and
\(y=r^\sharp(\eta)\), we have
\begin{equation*}
        \gamma(x,y)
        =
        \langle \xi,r^\sharp(\eta)\rangle
        =
        \langle \eta,r^\sharp(\xi)\rangle
        =
        \gamma(y,x).
\end{equation*}
Thus \(\gamma\) is symmetric.

It remains to prove that \(\gamma\) is a \(u\)-generalized \(2\)-cocycle. Set \(T=r^\sharp\). Then \(r=\sum_iT(e^i)\otimes e_i\), and by Theorem~\ref{thm:generalized-O-operator-S-equation}, \(T\) satisfies
\begin{equation*}
\begin{aligned}
T(\xi)\cdot T(\eta)
={}&
T\bigl(
\operatorname{ad}_{T(\xi)}^*\eta
-
R_{T(\eta)}^*\xi
\bigr)
+
\langle \xi,u\rangle T(\eta)
-
\langle \xi,T(\eta)\rangle u
\end{aligned}
\end{equation*}
for all \(\xi,\eta\in A^*\). Let \(a=T(\xi)\), \(b=T(\eta)\), and
\(c=T(\zeta)\). Using the definitions of \(\operatorname{ad}^*\) and \(R^*\),
we get
\begin{equation*}
\begin{aligned}
\gamma(a\cdot b,c)
={}&
-\gamma(b,[a,c])
+\gamma(a,c\cdot b)
+\gamma(a,u)\gamma(b,c)
-\gamma(a,b)\gamma(u,c).
\end{aligned}
\end{equation*}
By interchanging \(a\) and \(b\) in this identity and subtracting, we get
\begin{equation*}
\begin{aligned}
&\gamma(a\cdot b,c)-\gamma(a,b\cdot c)
-\gamma(b\cdot a,c)+\gamma(b,a\cdot c) \\
={}&
\gamma(a,u)\gamma(b,c)-\gamma(b,u)\gamma(a,c).
\end{aligned}
\end{equation*}
Therefore
\begin{equation*}
\begin{aligned}
&\gamma(a\cdot b,c)-\gamma(a,b\cdot c)
-\gamma(b\cdot a,c)+\gamma(b,a\cdot c) \\
&\quad
-\gamma(a,u)\gamma(b,c)
+\gamma(b,u)\gamma(a,c)=0 .
\end{aligned}
\end{equation*}
Since \(T=r^\sharp\) is surjective, this identity holds for all \(a,b,c\in A\). Hence \(\gamma\) is a nondegenerate \(u\)-generalized \(2\)-cocycle, and \((A,\cdot,\gamma,u)\) is a \(u\)-generalized Hessian pre-Lie algebra.
\end{proof}

\subsection{Two Types of Symmetric Solutions}

Let \((A,\cdot)\) be a pre-Lie algebra, let \(0\neq u\in \operatorname{Ann}(A)\), and let
\[
        r=\sum_i x_i\otimes y_i\in S^2(A)
\]
be a symmetric solution of the \(u\)-generalized \(S\)-equation. Equivalently, this equation can be written in the following expanded tensor form:
\begin{equation}\label{eq:u-S-expanded-type}
\begin{aligned}
0
=
\sum_{i,j}\Big(
&-x_i\cdot x_j\otimes y_i\otimes y_j
+x_i\otimes y_i\cdot x_j\otimes y_j
+x_i\otimes x_j\otimes [y_i,y_j]
\Big)  \\
&-\sum_i u\otimes x_i\otimes y_i
+\sum_i x_i\otimes u\otimes y_i ,
\end{aligned}
\end{equation}
where the bracket is taken in \(A_{\mathrm c}\).

Define the linear map \(r^\sharp:A^*\to A\) from Proposition~\ref{thm:nondegenerate-S-solution-generalized-Hessian} by
\begin{equation}\label{e4}
        r^\sharp(\xi)
        =
        \sum_i\langle \xi,y_i\rangle x_i
        =
        \sum_i\langle \xi,x_i\rangle y_i
\end{equation}
for all \(\xi\in A^*\), and let \(H=\operatorname{Im}(r^\sharp)\). Then
\(H\) is a subspace of \(A\) and \(r\in H\otimes H\). Indeed, let \(\dim H=k\), and let \(e_1,\ldots,e_k\) be a basis of \(H\) extended to a
basis \(e_1,\ldots,e_n\) of \(A\). Then \(r\) can be written in the form
\begin{equation*}
        r=\sum_{p,q=1}^k \alpha_{pq}e_p\otimes e_q
        =
        \sum_{q=1}^k a_q\otimes e_q,
        \qquad a_q=r^\sharp(e^q)\in H .
\end{equation*}
And \(H=\operatorname{span}\{a_1,\ldots,a_k\}\).

Let \(e_1^*,\ldots,e_n^*\) be the dual basis of A. Applying \(\operatorname{id}\otimes e^p\otimes e^q\), \(1\leq p,q\leq k\), to \eqref{eq:u-S-expanded-type}, we obtain
\[
r^\sharp(e_p^*)\cdot r^\sharp(e_q^*)\in H+\langle u\rangle
\]
Therefore
\[
        H\cdot H\subseteq H+\langle u\rangle .
\]
Since \(u\in\operatorname{Ann}(A)\), it follows that
\[
        H_1:=H+\langle u\rangle
\]
is a pre-Lie subalgebra of \(A\).
Thus two cases naturally arise:
\begin{enumerate}
    \item If \(u\in H\), then \(H\) itself is a pre-Lie subalgebra of
    \(A\). We call such symmetric solutions \emph{solutions of type \(1\)}.

    \item If \(u\notin H\), then \(H_1\neq H\), and \(H_1=H\oplus \langle u\rangle\) is a
    pre-Lie subalgebra of \(A\). We call such symmetric solutions
    \emph{solutions of type \(2\)}.
\end{enumerate}

\begin{ex}
Let \(A=\operatorname{span}\{x,y,u\}\) be the \(3\)-dimensional pre-Lie algebra with nonzero products \(x\cdot y=y\cdot x=-u\). Then \(u\in\operatorname{Ann}(A)\), and \(r=x\otimes y+y\otimes x\) is a solution of the \(u\)-generalized \(S\)-equation. Indeed, the classical \(S\)-equation part is cancelled exactly by the \(u\)-correction terms.

Moreover, the map \(r^\sharp:A^*\to A\) associated with \(r\) satisfies
\begin{equation*}
        \operatorname{Im}(r^\sharp)=\operatorname{span}\{x,y\}.
\end{equation*}
Since \(x\cdot y=-u\notin \operatorname{span}\{x,y\}\), this subspace is not a pre-Lie subalgebra of \(A\). Hence \(r\) is a symmetric solution of
type \(2\).
\end{ex}

\begin{thm}\label{prop:type1-LSA}
Let \((A,\cdot)\) be a pre-Lie algebra and let \(0\neq u\in\operatorname{Ann}(A)\). Then symmetric solutions of the \(u\)-generalized \(S\)-equation of type \(1\) are in one-to-one correspondence with \(u\)-generalized Hessian pre-Lie subalgebras \((H,\cdot,\gamma,u)\) of \(A\).
\end{thm}

\begin{proof}
Let \(r=\sum_i x_i\otimes y_i\in S^2(A)\) be a symmetric solution of the
\(u\)-generalized \(S\)-equation of type \(1\). Define \(r^\sharp:A^*\to A\) by \eqref{e4} and set \(H=\operatorname{Im}(r^\sharp)\). Since \(r\) is of type \(1\), we have \(u\in H\). Hence \(H\) is a pre-Lie subalgebra of \(A\). Moreover, \(r\in S^2(H)\).

Let $H^\circ=\{\xi\in A^*\mid \xi(H)=0\}.$ Since \(r\) is symmetric, \(\ker r^\sharp=H^\circ\). Indeed, for any \(\xi,\alpha\in A^*\), we have
\begin{equation*}
        \langle \alpha,r^\sharp(\xi)\rangle
        =
        \langle \alpha\otimes\xi,r\rangle
        =
        \langle \xi\otimes\alpha,r\rangle
        =
        \langle \xi,r^\sharp(\alpha)\rangle .
\end{equation*}
Thus \(\xi\) annihilates \(\operatorname{Im}(r^\sharp)\) if and only if
\(r^\sharp(\xi)=0\). Therefore \(r^\sharp\) induces an isomorphism
\begin{equation*}
        \overline{r^\sharp}:H^*\cong A^*/H^\circ\to H.
\end{equation*}
Now regard \(r\) as an element of \(S^2(H)\). Since \(\overline{r^\sharp}\) is an isomorphism, \(r\) is nondegenerate on \(H\). As \(H\) is a pre-Lie subalgebra of \(A\) and \(u\in H\), the \(u\)-generalized \(S\)-equation for \(r\) in \(A\) is exactly the same equation in \(H\). 

Define
\begin{equation*}
        \gamma(a,b)
        =
        \left\langle (\overline{r^\sharp})^{-1}(a),b\right\rangle
\end{equation*}
for all \(a,b\in H\). By Proposition~\ref{thm:nondegenerate-S-solution-generalized-Hessian}, applied to the pre-Lie algebra \(H\), the quadruple \((H,\cdot,\gamma,u)\) is a \(u\)-generalized Hessian pre-Lie algebra. Thus every type \(1\) solution gives a \emph{\(u\)-generalized Hessian pre-Lie subalgebra} of \(A\).

Conversely, let \((H,\cdot,\gamma,u)\) be a \emph{\(u\)-generalized Hessian
pre-Lie subalgebra} of \(A\). By Proposition~\ref{prop:generalized-Hessian-O-operator}, \(T_\gamma=(\gamma^\sharp)^{-1}:H^*\to H\) is a symmetric \emph{\(u\)-generalized relative Rota--Baxter operator associated with the coregular representation} of \(H\). Hence, by Theorem~\ref{thm:generalized-O-operator-S-equation}, for any basis \(\{e_i\}\) of \(H\) with dual basis \(\{e_i^*\}\),
\begin{equation*}
        r_\gamma=\sum_i T_\gamma(e_i^*)\otimes e_i
        \in S^2(H)
\end{equation*}
is a symmetric solution of the \(u\)-generalized \(S\)-equation on
\(H\). Since \(H\) is a pre-Lie subalgebra of \(A\) and \(u\in H\), the
same tensor is also a solution on \(A\). Moreover, the associated map \(r_\gamma^\sharp:A^*\to A\) is given by
\begin{equation*}
        r_\gamma^\sharp(\xi)=T_\gamma(\xi|_H)
\end{equation*}
for all \(\xi\in A^*\). Since the restriction map \(A^*\to H^*\) is surjective, we get
\begin{equation*}
        \operatorname{Im}(r_\gamma^\sharp)=H .
\end{equation*}
Since \(u\in H\), the resulting solution is of type \(1\). This proves the correspondence.
\end{proof}

\begin{thm}\label{prop:type2-LSA}
Let \((A,\cdot)\) be a pre-Lie algebra and let \(0\neq u\in \operatorname{Ann}(A)\). Then symmetric solutions of the \(u\)-generalized \(S\)-equation of type \(2\) are in one-to-one correspondence with pairs \((H,\gamma)\), where \(H\subseteq A\) is a subspace and \(\gamma\in S^2(H^*)\) is nondegenerate, satisfying the following conditions:
\begin{enumerate}
    \item \(u\notin H\);
    \item \(x*y:=x\cdot y+\gamma(x,y)u\in H\) for all \(x,y\in H\);
    \item \((H,*,\gamma)\) is a Hessian pre-Lie algebra.
\end{enumerate}
\end{thm}

\begin{proof}
Let \(r=\sum_i x_i\otimes y_i\in S^2(A)\) be a type \(2\) symmetric solution,
define \(r^\sharp:A^*\to A\) by \eqref{e4} and set \(H=\operatorname{Im}(r^\sharp)\). As in the proof of Theorem~\ref{prop:type1-LSA}, \(r^\sharp\) induces an isomorphism \(\overline{r^\sharp}:H^*\cong A^*/H^\circ\to H\), where \(H^\circ=\{\xi\in A^*\mid \xi(H)=0\}\). Define
\begin{equation}\label{eq:type2-omega-LSA}
        \gamma(a,b)=\bigl\langle (\overline{r^\sharp})^{-1}(a),b\bigr\rangle
\end{equation}
for all \(a,b\in H\). Then \(\gamma\) is symmetric and nondegenerate.

Let \(f,g\in A^*\), and put \(a=r^\sharp(f)\), \(b=r^\sharp(g)\). By the \(u\)-generalized relative Rota--Baxter operators identity associated with the
\(u\)-generalized \(S\)-equation, we have
\begin{equation*}
\begin{aligned}
a\cdot b
={}&
r^\sharp\bigl(
\operatorname{ad}_a^*g
-
R_b^*f
\bigr)
+
\langle f,u\rangle b
-
\langle f,b\rangle u .
\end{aligned}
\end{equation*}
Since \(\langle f,b\rangle=\gamma(a,b)\), it follows that \(a\cdot b+\gamma(a,b)u\in H\). Thus the product \(a*b:=a\cdot b+\gamma(a,b)u\) is well defined on \(H\). Since \(u\in\operatorname{Ann}(A)\), the space \(H\oplus\langle u\rangle\) is a pre-Lie subalgebra of \(A\), and the original product on this subalgebra is given by
\begin{equation*}
        a\cdot b=a*b-\gamma(a,b)u,
        \qquad
        u\cdot (H\oplus\langle u\rangle)
        =
        (H\oplus\langle u\rangle)\cdot u=0 .
\end{equation*}
In particular, \((H,*)\) is the quotient pre-Lie algebra \((H\oplus\langle u\rangle)/\langle u\rangle\).

Substituting \(a\cdot b=a*b-\gamma(a,b)u\) into the \(u\)-generalized \(S\)-equation, the \(u\)-terms cancel by the inverse-tensor identities for \(r\) and \(\gamma\). Hence \(r\) viewed as an element of \(S^2(H)\), satisfies the ordinary \(S\)-equation in the pre-Lie algebra \((H,*)\). Since \(r\) is nondegenerate on \(H\), the standard correspondence between nondegenerate symmetric solutions of the classical \(S\)-equation and Hessian pre-Lie algebras implies that \((H,*,\gamma)\) is a Hessian pre-Lie algebra.

Conversely, suppose that a pair \((H,\gamma)\) satisfies the three stated
conditions. Let \(e_1,\ldots,e_k\) be a basis of \(H\), and let \(C=(c_{ij})\) be the inverse of the matrix \((\gamma(e_i,e_j))\). Define
\begin{equation*}
        r=\sum_{i,j=1}^k c_{ij}e_i\otimes e_j\in S^2(H)\subseteq S^2(A).
\end{equation*}
Thus \(r\) is the inverse tensor of \(\gamma\). In particular, for every
\(a\in H\),
\begin{equation}\label{eq:type2-inverse-identity-LSA}
        \sum_{i,j=1}^k c_{ij}\gamma(e_j,a)e_i=a .
\end{equation}
Since \((H,*,\gamma)\) is a Hessian pre-Lie algebra, the tensor \(r\) is a symmetric solution of the classical \(S\)-equation of \((H,*)\). Now use $x\cdot y=x*y-\gamma(x,y)u$ for all \(x,y\in H\). Because \(\gamma\) is symmetric, the sub-adjacent Lie brackets determined by \(\cdot\) and by \(*\) agree on \(H\). Hence the classical \(S\)-equation part with respect to \(*\) vanishes, and the additional terms produced by \(-\gamma(x,y)u\) are exactly cancelled by the correction terms
\begin{equation*}
        -(u\otimes1\otimes1)r_{23}
        +(1\otimes u\otimes1)r_{13}
\end{equation*}
in the \(u\)-generalized \(S\)-equation. Indeed, \eqref{eq:type2-inverse-identity-LSA} implies the tensor identities
\begin{equation*}
        \sum_{p,q}\gamma(x_p,x_q)y_p\otimes y_q=r,
        \qquad
        \sum_{p,q}\gamma(y_p,x_q)x_p\otimes y_q=r,
\end{equation*}
where \(r=\sum_p x_p\otimes y_p\). Thus
\begin{equation*}
        -r_{12}\cdot r_{13}
        +r_{12}\cdot r_{23}
        +[r_{13},r_{23}]
        -(u\otimes1\otimes1)r_{23}
        +(1\otimes u\otimes1)r_{13}
        =
        0 .
\end{equation*}
Since \(u\notin H\), the solution is of type \(2\). 
\end{proof}

\begin{prop}\label{prop:shift-type2-LSA}
Let \((A,\cdot)\) be a pre-Lie algebra, let \(0\neq u\in\operatorname{Ann}(A)\), and let \(r\in S^2(A)\) be a symmetric solution of the classical \(S\)-equation. Let \((H,\cdot,\gamma)\) be the Hessian pre-Lie subalgebra of \(A\) corresponding to \(r\), and assume that \(u\notin H\). Let \(e_1,\ldots,e_k\) be a basis of \(H\), with \(e_i\cdot e_j=\sum_{s=1}^k\kappa_{ij}^s e_s\). Suppose that the linear system
\begin{equation}\label{eq:shift-type2-LSA-system}
        \sum_{s=1}^k\kappa_{ij}^s\lambda_s=\gamma(e_i,e_j)
\end{equation}
has a solution \((\lambda_1,\ldots,\lambda_k)\in F^k\) for all
\(1\le i,j\le k\). If $r=\sum_{i,j=1}^k\alpha_{ij}e_i\otimes e_j,$ then
\begin{equation*}
        r_\lambda
        =
        \sum_{i,j=1}^k
        \alpha_{ij}(e_i+\lambda_i u)\otimes(e_j+\lambda_j u)
\end{equation*}
is a symmetric solution of the \(u\)-generalized \(S\)-equation of type \(2\).
\end{prop}

\begin{proof}
Define \(\Phi:H\to A\) by \(\Phi(e_i)=e_i+\lambda_i u\). Since \(u\in\operatorname{Ann}(A)\), we have
\begin{equation}\label{eq:shift-type2-LSA-product}
\begin{aligned}
\Phi(e_i)\cdot\Phi(e_j)
&=e_i\cdot e_j
=\sum_{s=1}^k\kappa_{ij}^s e_s                                      \\
&=
\sum_{s=1}^k\kappa_{ij}^s\Phi(e_s)
-
\left(\sum_{s=1}^k\kappa_{ij}^s\lambda_s\right)u                     \\
&=
\Phi(e_i\cdot e_j)-\gamma(e_i,e_j)u .
\end{aligned}
\end{equation}
Define a product on \(\Phi(H)\) by \(\Phi(x)*\Phi(y)=\Phi(x\cdot y)\) for all \(x,y\in H\), and define a symmetric bilinear form \(\gamma_\Phi\) on \(\Phi(H)\) by \(\gamma_\Phi(\Phi(a),\Phi(b))=\gamma(a,b)\) for all \(a,b\in H\). Then \(\Phi:(H,\cdot,\gamma)\to(\Phi(H),*,\gamma_\Phi)\) is an isomorphism of Hessian pre-Lie algebras. 

Moreover, \eqref{eq:shift-type2-LSA-product} gives
\begin{equation*}
        \Phi(x)\cdot\Phi(y)
        =
        \Phi(x)*\Phi(y)
        -
        \gamma_\Phi(\Phi(x),\Phi(y))u.
\end{equation*}
Also \(u\notin\Phi(H)\), because \(\Phi(H)\) is the graph of a linear map
\(H\to \langle u\rangle\) inside \(H\oplus\langle u\rangle\), while \(u\notin H\). Therefore Theorem~\ref{prop:type2-LSA} applies, and
\begin{equation*}
        r_\lambda
        =
        \sum_{i,j=1}^k
        \alpha_{ij}\Phi(e_i)\otimes\Phi(e_j)
\end{equation*}
is a symmetric solution of the \(u\)-generalized \(S\)-equation of type \(2\).
\end{proof}

\begin{ex}
Let \(A=\operatorname{span}\{e_1,e_2,e_3,u\}\) be the \(4\)-dimensional pre-Lie algebra with nonzero products
\begin{equation*}
        e_1\cdot e_1=e_1,\qquad
        e_1\cdot e_2=e_2\cdot e_1=e_2,\qquad
        e_1\cdot e_3=e_3\cdot e_1=e_3,\qquad
        e_2\cdot e_2=e_3 .
\end{equation*}
Then \(A\) is an associative left-symmetric algebra, and \(\operatorname{Ann}(A)=\langle u\rangle\). Let 
\begin{equation*}
      H=\operatorname{span}\{e_1,e_2,e_3\}.
\end{equation*}
Then \(H\) is a pre-Lie subalgebra of \(A\). Define a symmetric bilinear form
\(\gamma\) on \(H\) by
\begin{equation*}
        \gamma(e_1,e_3)=1,\qquad
        \gamma(e_2,e_2)=1,
\end{equation*}
and all other values to be zero. Then \((H,\cdot,\gamma)\) is a Hessian pre-Lie subalgebra of \(A\). The inverse tensor of \(\gamma\) is
\begin{equation*}
        r=e_1\otimes e_3+e_2\otimes e_2+e_3\otimes e_1,
\end{equation*}
which is a symmetric solution of the classical \(S\)-equation of \(A\).

With respect to the basis \(e_1,e_2,e_3\) of \(H\), the linear system \(\sum_{s=1}^3\kappa_{ij}^s\lambda_s=\gamma(e_i,e_j)\) has the solution \((\lambda_1,\lambda_2,\lambda_3)=(0,0,1)\). Thus, setting
\begin{equation*}
        \Phi(e_1)=e_1,\qquad
        \Phi(e_2)=e_2,\qquad
        \Phi(e_3)=e_3+u,
\end{equation*}
Proposition~\ref{prop:shift-type2-LSA} gives that
\begin{equation*}
        r_\lambda
        =
        e_1\otimes(e_3+u)
        +e_2\otimes e_2
        +(e_3+u)\otimes e_1
\end{equation*}
is a symmetric solution of the \(u\)-generalized \(S\)-equation of type \(2\) of the pre-Lie algebra \(A\).
\end{ex}

\subsection{Factorizable Solutions}

We study factorizable solutions of the nontrivial \(A\)-generalized Yang--Baxter equation. For \(r=\sum_i x_i\otimes y_i\in A\otimes A\), define linear maps \(r_+,r_-:A^*\to A\) by
\begin{equation*}
        r_+(\xi)=\sum_i\langle \xi,x_i\rangle y_i,
        \qquad
        r_-(\xi)=r^\sharp(\xi)=\sum_i\langle \xi,y_i\rangle x_i
\end{equation*}
for all \(\xi\in A^*\), and set \(I=r_+-r_-\). A key tool is the following lemma.

\begin{lem}\cite{WangBaiLiuSheng2024}\label{lem:factorizable-r-minus}
Let \((A,\cdot)\) be a pre-Lie algebra and let \(r\in A\otimes A\). Assume that \(I=r_+-r_-\) is a linear isomorphism and that \(r-\sigma(r)\) is \((L,\operatorname{ad})\)-invariant. Let \(\diamond\) be the product on \(A^*\) defined by $\langle \xi\diamond\eta,x\rangle = \langle \xi\otimes\eta,\alpha_r(x)\rangle$ for all \(\xi,\eta\in A^*\) and \(x\in A\). Then \(r\) is a factorizable solution of the classical Yang--Baxter equation of \(A\) if and only if
\begin{equation*}
        r_-(\xi\diamond\eta)=r_-(\xi)\cdot r_-(\eta)
\end{equation*}
for all \(\xi,\eta\in A^*\).
\end{lem}

The preceding lemma leads to a characterization of factorizable solutions of
the classical Yang--Baxter equation of a pre-Lie algebra in terms of quadratic Rota--Baxter pre-Lie algebras of nonzero weight. We first recall the following definitions.

\begin{defn}\cite{Li}
Let \((A,\cdot)\) be a pre-Lie algebra and let \(\lambda\in F\). A linear map \(\mathcal R:A\to A\) is called a Rota--Baxter operator of weight \(\lambda\) if
\begin{equation*}
        \mathcal R(x)\cdot \mathcal R(y)
        =
        \mathcal R\bigl(
        \mathcal R(x)\cdot y
        +x\cdot \mathcal R(y)
        +\lambda x\cdot y
        \bigr)
\end{equation*}
for all \(x,y\in A\).
\end{defn}

\begin{defn}\cite{WangBaiLiuSheng2024}
A Rota--Baxter pre-Lie algebra of weight \(\lambda\) is a triple \((A,\cdot,\mathcal R)\), where \((A,\cdot)\) is a pre-Lie algebra and
\(\mathcal R:A\to A\) is a Rota--Baxter operator of weight \(\lambda\).
\end{defn}

\begin{defn}\cite{WangBaiLiuSheng2024}
Let \((A,\cdot,\mathcal R)\) be a Rota--Baxter pre-Lie algebra of weight \(\lambda\), and let \(S\) be a nondegenerate skew-symmetric bilinear form on \(A\). The quadruple \((A,\cdot,\mathcal R,S)\) is called a quadratic Rota--Baxter pre-Lie algebra of weight \(\lambda\) if \((A,\cdot,S)\) is a quadratic pre-Lie algebra and
\begin{equation*}
        S(x,\mathcal R(y))
        +
        S(\mathcal R(x),y)
        +
        \lambda S(x,y)
        =
        0
\end{equation*}
for all \(x,y\in A\).
\end{defn}

\begin{thm}\cite{WangBaiLiuSheng2024}
Let \((A,\cdot,\mathcal R,S)\) be a quadratic Rota--Baxter pre-Lie algebra of nonzero weight \(\lambda\). Let \(I_S:A^*\to A\) be the linear isomorphism determined by \(S(I_S(\xi),y)=\langle \xi,y\rangle\) for all \(\xi\in A^*\) and \(y\in A\).
Define
\begin{equation*}
        r_-=\frac{1}{\lambda}\mathcal R\circ I_S:A^*\to A,
\end{equation*}
and let \(r\in A\otimes A\) be the unique tensor satisfying \(\langle \operatorname{id}\otimes \xi,r\rangle=r_-(\xi)\) for all \(\xi\in A^*\). Then
\begin{equation*}
        r_-(\xi\diamond\eta)=r_-(\xi)\cdot r_-(\eta)
\end{equation*}
for all \(\xi,\eta\in A^*\). Consequently, \(r\) is a factorizable solution
of the classical Yang--Baxter equation of \(A\).
\end{thm}

\begin{thm}\cite{WangBaiLiuSheng2024}
Let \((A,\cdot)\) be a pre-Lie algebra, and let \(r\in A\otimes A\) be a factorizable solution of the classical Yang--Baxter equation of \(A\). Let
\(0\neq\lambda\in F\). Set \(I=r_+-r_-:A^*\to A\), and define
\begin{equation*}
        \mathcal R=\lambda r_-\circ I^{-1}:A\to A,
        \qquad
        S_I(x,y)=\langle I^{-1}(x),y\rangle
\end{equation*}
for all \(x,y\in A\). Then \((A,\cdot,\mathcal R,S_I)\) is a quadratic Rota--Baxter pre-Lie algebra of weight \(\lambda\).
\end{thm}

We now generalize the above result to obtain factorizable solutions of the
nontrivial \(A\)-generalized Yang--Baxter equation. Let \(0\neq u\in\operatorname{Ann}(A)\) and \(r\in A\otimes A\). Define the
product \(\diamond_u\) on \(A^*\) by
\begin{equation*}
        \langle \xi\diamond_u\eta,x\rangle
        =
        \langle \xi\otimes\eta,\alpha_r(x)\rangle
        +
        \langle \xi,u\rangle\langle \eta,x\rangle
\end{equation*}
for all \(\xi,\eta\in A^*\) and \(x\in A\). A key tool for studying factorizable solutions of the \(A\)-generalized Yang--Baxter equation
is the following \(u\)-analogue of Lemma~\ref{lem:factorizable-r-minus}.

\begin{lem}\label{lem:u-factorizable-r-minus}
Let \((A,\cdot)\) be a pre-Lie algebra, let \(u\in\operatorname{Ann}(A)\), and let \(r\in A\otimes A\). Assume that \(r_+-r_-:A^*\to A\) is a linear isomorphism and that \(r-\sigma(r)\) is \((L,\operatorname{ad})\)-invariant. Then \(r\) is a factorizable solution of the \(u\)-generalized Yang--Baxter equation if and only if
\begin{equation*}
        r_-(\xi\diamond_u\eta)
        =
        r_-(\xi)\cdot r_-(\eta)
        +
        \langle \xi\otimes\eta,r\rangle u
\end{equation*}
for all \(\xi,\eta\in A^*\).
\end{lem}

\begin{proof}
Let \(\diamond\) be the product on \(A^*\) obtained from \(\diamond_u\) by
setting \(u=0\). The computation underlying Lemma~\ref{lem:factorizable-r-minus} gives
\begin{equation*}
        \langle \xi\otimes\operatorname{id}\otimes\eta,[[r,r]]\rangle
        =
        r_-(\xi\diamond\eta)-r_-(\xi)\cdot r_-(\eta)
\end{equation*}
for all \(\xi,\eta\in A^*\). Using
\begin{equation*}
        [[r,r]]_u
        =
        [[r,r]]
        +
        \sum_i
        \bigl(
        u\otimes x_i\otimes y_i
        -
        x_i\otimes u\otimes y_i
        \bigr),
\end{equation*}
we get
\begin{equation*}
\begin{aligned}
\langle \xi\otimes\operatorname{id}\otimes\eta,[[r,r]]_u\rangle
={}&
r_-(\xi\diamond\eta)-r_-(\xi)\cdot r_-(\eta) +
\langle \xi,u\rangle r_-(\eta)
-\langle \xi\otimes\eta,r\rangle u .
\end{aligned}
\end{equation*}
Since
\(\xi\diamond_u\eta=\xi\diamond\eta+\langle\xi,u\rangle\eta\), it follows that
\begin{equation*}
        \langle \xi\otimes\operatorname{id}\otimes\eta,[[r,r]]_u\rangle
        =
        r_-(\xi\diamond_u\eta)
        -
        r_-(\xi)\cdot r_-(\eta)
        -
        \langle \xi\otimes\eta,r\rangle u .
\end{equation*}
By the nondegeneracy of the natural pairing, \([[r,r]]_u=0\) is equivalent to
the stated identity. Together with the assumed invertibility of \(r_+-r_-\)
and the \((L,\operatorname{ad})\)-invariance of \(r-\sigma(r)\), this is precisely the factorizable solution condition.
\end{proof}

We now introduce \emph{\(u\)-generalized quadratic Rota--Baxter pre-Lie algebras of weight \(\lambda\)}.

\begin{defn}
Let \((A,\cdot)\) be a pre-Lie algebra, let \(0\neq u\in\operatorname{Ann}(A)\), \(\lambda\in F\), and let \(S\) be a bilinear form on \(A\). A linear map \(\mathcal R:A\to A\) is called a \emph{\(u\)-generalized Rota--Baxter operator of weight \(\lambda\)} if
\begin{equation*}
        \mathcal R(x)\cdot \mathcal R(y)
        =
        \mathcal R\bigl(
        \mathcal R(x)\cdot y
        +x\cdot \mathcal R(y)
        +\lambda x\cdot y
        +\lambda S(x,u)y
        \bigr)
        -
        \lambda S(x,\mathcal R(y))u
\end{equation*}
for all \(x,y\in A\).
\end{defn}

\begin{defn}
Let \((A,\cdot,S)\) be a quadratic pre-Lie algebra, let \(0\neq u\in\operatorname{Ann}(A)\), and \(\lambda\in F\). A quintuple \((A,\cdot,\mathcal R,S,u)\) is called a \emph{\(u\)-generalized quadratic Rota--Baxter pre-Lie algebra of weight \(\lambda\)} if \(\mathcal R:A\to A\) is a \emph{\(u\)-generalized Rota--Baxter operator of weight \(\lambda\)} and
\begin{equation*}
        S(x,\mathcal R(y))
        +
        S(\mathcal R(x),y)
        +
        \lambda S(x,y)
        =
        0
\end{equation*}
for all \(x,y\in A\).
\end{defn}

The following result constructs factorizable solutions of the \(u\)-generalized Yang--Baxter equation from \(u\)-generalized quadratic Rota--Baxter pre-Lie algebras of nonzero weight.

\begin{thm}\label{thm:generalized-quadratic-RB-to-u-factorizable}
Let \((A,\cdot,\mathcal R,S,u)\) be a \(u\)-generalized quadratic Rota--Baxter pre-Lie algebra of nonzero weight \(\lambda\). Let \(I_S:A^*\to A\) be the linear isomorphism determined by \(S(I_S(\xi),y)=\langle \xi,y\rangle\) for all \(\xi\in A^*\) and \(y\in A\). Define
\begin{equation*}
        r_-=\frac{1}{\lambda}\mathcal R\circ I_S:A^*\to A,
\end{equation*}
and let \(r\in A\otimes A\) be the unique tensor satisfying \(\langle \operatorname{id}\otimes \xi,r\rangle=r_-(\xi)\) for all \(\xi\in A^*\). Then \(r\) is a factorizable solution of the \(u\)-generalized Yang--Baxter equation.
\end{thm}

\begin{proof}
Let \(x=I_S(\xi)\) and \(y=I_S(\eta)\). The identity
\begin{equation*}
        S(x,\mathcal R(y))+S(\mathcal R(x),y)+\lambda S(x,y)=0
\end{equation*}
implies \(r_+-r_-=I_S\). Indeed,
\begin{equation*}
\begin{aligned}
\langle \eta,(r_+-r_-)(\xi)\rangle
={}&
\frac{1}{\lambda}
\bigl(S(x,\mathcal R(y))-S(y,\mathcal R(x))\bigr) = S(y,x) =\langle \eta,I_S(\xi)\rangle .
\end{aligned}
\end{equation*}
Hence \(r_+-r_-\) is a linear isomorphism. Moreover, the invariance of \(S\)
implies that \(r-\sigma(r)\) is \((L,\operatorname{ad})\)-invariant.

Let \(\diamond\) denote the classical product on \(A^*\) corresponding to
\(u=0\). By the classical relation between quadratic Rota--Baxter pre-Lie algebras and factorizable tensors, we have
\begin{equation*}
        I_S(\xi\diamond\eta)
        =
        \frac{1}{\lambda}
        \bigl(
        x\cdot\mathcal R(y)
        +\lambda x\cdot y
        +\mathcal R(x)\cdot y
        \bigr).
\end{equation*}
Since
\(\xi\diamond_u\eta=\xi\diamond\eta+\langle \xi,u\rangle\eta\), we get
\begin{equation*}
        I_S(\xi\diamond_u\eta)
        =
        \frac{1}{\lambda}
        \bigl(
        x\cdot\mathcal R(y)
        +\lambda x\cdot y
        +\mathcal R(x)\cdot y
        +\lambda S(x,u)y
        \bigr).
\end{equation*}
Thus
\begin{equation*}
        r_-(\xi\diamond_u\eta)
        =
        \frac{1}{\lambda^2}
        \mathcal R\bigl(
        x\cdot\mathcal R(y)
        +\lambda x\cdot y
        +\mathcal R(x)\cdot y
        +\lambda S(x,u)y
        \bigr).
\end{equation*}
On the other hand,
\begin{equation*}
\begin{aligned}
r_-(\xi)\cdot r_-(\eta)
+\langle \xi\otimes\eta,r\rangle u
={}&
\frac{1}{\lambda^2}\mathcal R(x)\cdot\mathcal R(y)
+
\frac{1}{\lambda}S(x,\mathcal R(y))u  \\
={}&
\frac{1}{\lambda^2}
\bigl(
\mathcal R(x)\cdot\mathcal R(y)
+\lambda S(x,\mathcal R(y))u
\bigr).
\end{aligned}
\end{equation*}
By the defining identity of the \(u\)-generalized Rota--Baxter operator, the last two displayed expressions are equal. Hence
\begin{equation*}
        r_-(\xi\diamond_u\eta)
        =
        r_-(\xi)\cdot r_-(\eta)
        +
        \langle \xi\otimes\eta,r\rangle u .
\end{equation*}
By Lemma~\ref{lem:u-factorizable-r-minus}, \(r\) is a factorizable solution
of the \(u\)-generalized Yang--Baxter equation.
\end{proof}

Conversely, factorizable solutions of the \(u\)-generalized Yang--Baxter equation induce \(u\)-generalized quadratic Rota--Baxter pre-Lie algebras of
nonzero weight.

\begin{thm}\label{thm:u-factorizable-to-generalized-quadratic-RB}
Let \((A,\cdot)\) be a pre-Lie algebra, let \(0\neq u\in\operatorname{Ann}(A)\), and \(r\in A\otimes A\) be a factorizable solution of the \(u\)-generalized Yang--Baxter equation.
Let \(0\neq\lambda\in F\). Set \(I=r_+-r_-:A^*\to A\), and define
\begin{equation*}
        \mathcal R=\lambda r_-\circ I^{-1}:A\to A,
        \qquad
        S_I(x,y)=\langle I^{-1}(x),y\rangle
\end{equation*}
for all \(x,y\in A\). Then \((A,\cdot,\mathcal R,S_I,u)\) is a \(u\)-generalized quadratic Rota--Baxter pre-Lie algebra of weight \(\lambda\).
\end{thm}

\begin{proof}
Since \(r\) is factorizable, \(I=r_+-r_-\) is a linear isomorphism and \(r-\sigma(r)\) is \((L,\operatorname{ad})\)-invariant. Hence \(S_I\) is nondegenerate and skew-symmetric, and the invariance of \(r-\sigma(r)\) is
equivalent to the invariance of \(S_I\). Thus \((A,\cdot,S_I)\) is a quadratic pre-Lie algebra.

The compatibility condition follows as in the classical case. Indeed, for
\(x=I(\xi)\) and \(y=I(\eta)\), we have
\begin{equation*}
\begin{aligned}
&S_I(x,\mathcal R(y))+S_I(\mathcal R(x),y)+\lambda S_I(x,y)={}
\lambda\bigl(
\langle \xi,r_-(\eta)\rangle
-\langle \eta,r_-(\xi)\rangle
+\langle \xi,I(\eta)\rangle
\bigr)=0 .
\end{aligned}
\end{equation*}

Therefore it remains to prove that \(\mathcal R\) is a \(u\)-generalized Rota--Baxter operator of weight \(\lambda\).
For \(x=I(\xi)\) and \(y=I(\eta)\), the classical relation gives
\begin{equation*}
        I(\xi\diamond\eta)
        =
        r_-(\xi)\cdot I(\eta)
        +
        I(\xi)\cdot r_-(\eta)
        +
        I(\xi)\cdot I(\eta).
\end{equation*}
Since \(\xi\diamond_u\eta=\xi\diamond\eta+\langle \xi,u\rangle\eta\), we get
\begin{equation*}
        \lambda I(\xi\diamond_u\eta)
        =
        \mathcal R(x)\cdot y
        +
        x\cdot \mathcal R(y)
        +
        \lambda x\cdot y
        +
        \lambda S_I(x,u)y .
\end{equation*}
Applying \(\mathcal R\) to both sides gives
\begin{equation*}
\begin{aligned}
&\mathcal R\bigl(
\mathcal R(x)\cdot y+x\cdot\mathcal R(y)
+\lambda x\cdot y+\lambda S_I(x,u)y
\bigr)
-\lambda S_I(x,\mathcal R(y))u  \\
={}&
\lambda^2\bigl(
r_-(\xi\diamond_u\eta)-\langle \xi\otimes\eta,r\rangle u
\bigr).
\end{aligned}
\end{equation*}
By Lemma~\ref{lem:u-factorizable-r-minus}, the right-hand side equals \(\lambda^2 r_-(\xi)\cdot r_-(\eta)=\mathcal R(x)\cdot\mathcal R(y)\). Hence
\begin{equation*}
        \mathcal R(x)\cdot\mathcal R(y)
        =
        \mathcal R\bigl(
        \mathcal R(x)\cdot y+x\cdot\mathcal R(y)
        +\lambda x\cdot y+\lambda S_I(x,u)y
        \bigr)
        -
        \lambda S_I(x,\mathcal R(y))u .
\end{equation*}
Thus \(\mathcal R\) is a \(u\)-generalized Rota--Baxter operator of weight \(\lambda\), and the proof is complete.
\end{proof}

To construct \(u\)-generalized quadratic Rota--Baxter left-symmetric algebras of nonzero weight, we relate them to \(u\)-generalized Rota--Baxter symplectic Lie algebras of nonzero weight.

\begin{defn}
Let \((\mathfrak g,[\,,\,])\) be a Lie algebra, let \(0\ne u\in C(\mathfrak g)\), where \(C(\mathfrak g)\) denotes the center of \(\mathfrak g\), and let \(0\ne\lambda\in F\). A quintuple \((\mathfrak g,[\,,\,],\mathcal R,S,u)\) is called a \emph{\(u\)-generalized Rota--Baxter symplectic Lie algebra of weight \(\lambda\)} if the following conditions hold:
\begin{enumerate}
\item \((\mathfrak g,S)\) is a symplectic Lie algebra, namely \(S\) is a nondegenerate skew-symmetric bilinear form satisfying
\begin{equation*}
        S([x,y],z)+S([y,z],x)+S([z,x],y)=0
\end{equation*}
for all \(x,y,z\in\mathfrak g\).

\item The linear map \(\mathcal R:\mathfrak g\to\mathfrak g\) satisfies
\begin{equation*}
\begin{aligned}
[\mathcal R(x),\mathcal R(y)]
={}&
\mathcal R\bigl(
[\mathcal R(x),y]+[x,\mathcal R(y)]+\lambda[x,y]
\bigr) \\
&+
\lambda S(x,u)\bigl(\mathcal R(y)+\lambda y\bigr)
+\lambda S(u,y)\bigl(\mathcal R(x)+\lambda x\bigr)
\end{aligned}
\end{equation*}
for all \(x,y\in\mathfrak g\).

\item The compatibility condition
\begin{equation*}
        S(\mathcal R(x),y)+S(x,\mathcal R(y))+\lambda S(x,y)=0
\end{equation*}
holds for all \(x,y\in\mathfrak g\).
\end{enumerate}
\end{defn}

\begin{prop}\label{thm:u-QRB-LSA-symplectic-Lie}
Let \((A,\cdot,S)\) be a quadratic pre-Lie algebra, let \(0\neq u\in\operatorname{Ann}(A)\), and \(0\neq \lambda\in F\). Then \((A,\cdot,\mathcal R,S,u)\) is a \(u\)-generalized quadratic Rota--Baxter pre-Lie algebra of weight \(\lambda\) if and only if \((A_{\mathrm c},[\,,\,],\mathcal R,S,u)\) is a \(u\)-generalized Rota--Baxter symplectic Lie algebra of weight \(\lambda\).
\end{prop}

\begin{proof}
Since \((A,\cdot,S)\) is a quadratic pre-Lie algebra, \(S\) is a symplectic form on \(A_{\mathrm c}\). Moreover, \(u\in\operatorname{Ann}(A)\) implies \(u\in C(A_{\mathrm c})\). The compatibility condition between \(S\) and \(\mathcal R\) is identical in the two definitions. It remains to compare the operator identities.

For \(x,y,z\in A\), the quadratic identity of \(S\) gives
\begin{equation*}
\begin{aligned}
&S\Big(
\mathcal R(x)\cdot\mathcal R(y)
-\mathcal R\bigl(
\mathcal R(x)\cdot y+x\cdot\mathcal R(y)+\lambda x\cdot y
\bigr),
z
\Big)                                      \\
={}&
S\Big(
y,\,
\mathcal R\bigl(
[\mathcal R(x),z]+[x,\mathcal R(z)]+\lambda[x,z]
\bigr)
-[\mathcal R(x),\mathcal R(z)]
\Big).
\end{aligned}
\end{equation*}
Furthermore, the compatibility condition gives
\begin{equation*}
\begin{aligned}
&S\Big(
-\lambda S(x,u)\mathcal R(y)
+\lambda S(x,\mathcal R(y))u,\ z
\Big) ={}
S\Big(
y,\,
\lambda S(x,u)\bigl(\mathcal R(z)+\lambda z\bigr)
+\lambda S(u,z)\bigl(\mathcal R(x)+\lambda x\bigr)
\Big).
\end{aligned}
\end{equation*}
Adding these two identities and using the nondegeneracy of \(S\), we obtain
that the \(u\)-generalized Rota--Baxter identity on \((A,\cdot)\) is equivalent to
\begin{equation*}
\begin{aligned}
[\mathcal R(x),\mathcal R(z)]
={}&
\mathcal R\bigl(
[\mathcal R(x),z]+[x,\mathcal R(z)]+\lambda[x,z]
\bigr) \\
&+
\lambda S(x,u)\bigl(\mathcal R(z)+\lambda z\bigr)
+\lambda S(u,z)\bigl(\mathcal R(x)+\lambda x\bigr).
\end{aligned}
\end{equation*}
Renaming \(z\) as \(y\), this is precisely the operator identity in the definition of a \(u\)-generalized Rota--Baxter symplectic Lie algebra. Hence the two structures are equivalent.
\end{proof}

By the correspondence above, the study of \(u\)-generalized quadratic Rota--Baxter pre-Lie algebras of nonzero weight reduces to the study of \(u\)-generalized Rota--Baxter symplectic Lie algebras of nonzero weight. The latter class is highly restricted, as shown by the following result.

\begin{thm}\label{two-dimensionalabelian}
Let \(0\neq\lambda\in F\), \(0\neq u\in C(\mathfrak g)\), and let \((\mathfrak g,[\,,\,],\mathcal R,S,u)\) be a \(u\)-generalized Rota--Baxter symplectic Lie algebra of weight \(\lambda\). Then \(\mathfrak g\) is a two-dimensional abelian Lie algebra.
\end{thm}

\begin{proof}
Set \(P=\mathcal R+\frac{\lambda}{2}\operatorname{id}_{\mathfrak g}\) and
\(Q=P+\frac{\lambda}{2}\operatorname{id}_{\mathfrak g}\). The compatibility
condition gives
\begin{equation*}
        S(Px,y)+S(x,Py)=0
\end{equation*}
for all \(x,y\in\mathfrak g\). Thus \(P\) is skew-symmetric with respect to
\(S\).

In terms of \(P\) and \(Q\), the \emph{\(u\)-generalized Rota--Baxter} identity is equivalent to
\begin{equation*}
\begin{aligned}
T(x,y)
:={}&
[Px,Py]+\frac{\lambda^2}{4}[x,y]
-P\bigl([Px,y]+[x,Py]\bigr)        \\
={}&
\lambda\bigl(S(x,u)Qy-S(y,u)Qx\bigr).
\end{aligned}
\end{equation*}
We claim that \(\sum_{\mathrm{cyc}}S(T(x,y),z)=0\). Since \(P\) is skew-symmetric with respect to \(S\), we have
\begin{equation*}
\begin{aligned}
S(T(x,y),z)
={}&
S([Px,Py],z)
+\frac{\lambda^2}{4}S([x,y],z) + S([Px,y],Pz)+S([x,Py],Pz).
\end{aligned}
\end{equation*}
The cyclic sum of the second term vanishes because \(S\) is a Lie algebra
\(2\)-cocycle. The remaining cyclic sum is the sum of the cocycle identities
for the triples \((Px,Py,z)\), \((Py,Pz,x)\), and \((Pz,Px,y)\). Hence the
claim follows.

Let \(\varphi(x)=S(x,u)\). Using the expression for \(T(x,y)\), the above
claim gives
\begin{equation*}
        \sum_{\mathrm{cyc}}
        \bigl(
        \varphi(x)S(Qy,z)-\varphi(y)S(Qx,z)
        \bigr)=0 .
\end{equation*}
Since \(Q=P+\frac{\lambda}{2}\operatorname{id}_{\mathfrak g}\) and \(P\) is
skew-symmetric with respect to \(S\), we have
\begin{equation*}
        S(Qa,b)-S(Qb,a)=\lambda S(a,b)
\end{equation*}
for all \(a,b\in\mathfrak g\). Therefore
\begin{equation*}
        \varphi(x)S(y,z)+\varphi(y)S(z,x)+\varphi(z)S(x,y)=0
\end{equation*}
for all \(x,y,z\in\mathfrak g\), equivalently, \(\varphi\wedge S=0\).

Since \(u\neq0\) and \(S\) is nondegenerate, \(\varphi=S(\,\cdot\,,u)\) is
nonzero. If \(\dim\mathfrak g\ge4\), choose \(e\in\mathfrak g\) with \(\varphi(e)\neq0\). Then \(\langle u,e\rangle\) is a nondegenerate symplectic plane, and its symplectic orthogonal has positive even dimension.
Hence there exist \(w,z\in\langle u,e\rangle^\perp\) such that
\(S(w,z)\neq0\). But then
\begin{equation*}
        (\varphi\wedge S)(e,w,z)=\varphi(e)S(w,z)\neq0,
\end{equation*}
a contradiction. Thus \(\dim\mathfrak g<4\). Since \(\mathfrak g\) is symplectic and \(u\neq0\), we get \(\dim\mathfrak g=2\).

Finally, a two-dimensional Lie algebra with a nonzero central element is
abelian. Therefore \(\mathfrak g\) is a two-dimensional abelian Lie algebra.
\end{proof}

\section{\(u\)-Generalized Hessian pre-Lie algebras and Their Low-Dimensional Classification}\label{44}

In Section~\ref{33}, we introduced \(u\)-generalized Hessian pre-Lie algebras and established their connection with symmetric solutions of the \(u\)-generalized \(S\)-equation. In this section, we further study their structure.

\subsection{A Structure Theorem for \(u\)-Generalized Hessian pre-Lie algebras}

A natural question is how to construct \(u\)-generalized Hessian pre-Lie algebras. In the nontrivial case \(0\neq u\), we give a complete construction from Hessian pre-Lie algebras. The construction splits naturally into two cases according to whether \(u\) is isotropic with respect to \(\gamma\). If \(\gamma(u,u)\neq0\), the algebra is obtained by a
one-dimensional annihilator extension. If \(\gamma(u,u)=0\), it is obtained by a double extension. These two constructions exhaust all \(u\)-generalized Hessian pre-Lie algebras.

\begin{thm}\label{thm:nonisotropic-u-construction}
Let \((V,*,f)\) be a Hessian pre-Lie algebra, and let \(0\neq c\in F\).
Put \(\bar A=V\oplus \langle u\rangle\). Define a bilinear product \(\cdot\)
on \(\bar A\) by
\begin{equation*}
\begin{aligned}
        x\cdot y &= x*y-f(x,y)u, \quad \forall x,y\in V,\\
        u\cdot z &= z\cdot u=0, \quad \forall z\in\bar A .
\end{aligned}
\end{equation*}
Define a symmetric bilinear form \(\gamma\) on \(\bar A\) by
\begin{equation*}
\begin{aligned}
        \gamma(x,y) &= f(x,y), \quad \forall x,y\in V,\\
        \gamma(x,u) &= \gamma(u,x)=0, \quad \forall x\in V,\\
        \gamma(u,u) &= c .
\end{aligned}
\end{equation*}
Then \((\bar A,\cdot,\gamma,u)\) is a \(u\)-generalized Hessian pre-Lie algebra.
\end{thm}

\begin{proof}
The form \(\gamma\) is nondegenerate, and \(u\in\operatorname{Ann}(\bar A)\)
by construction. Since \(u\) is an annihilator, the pre-Lie identity only needs to be checked on \(V\). For \(x,y,z\in V\), the \(V\)-component is
the pre-Lie identity of \((V,*)\), while the \(u\)-component is
\begin{equation*}
        -f(x*y,z)+f(x,y*z)+f(y*x,z)-f(y,x*z),
\end{equation*}
which is zero because \(f\) is a Hessian \(2\)-cocycle.

It remains to check the \(u\)-generalized \(2\)-cocycle identity. For triples in \(V\), it reduces to the Hessian \(2\)-cocycle identity for \(f\).
If one of the first two variables is \(u\), the only nontrivial cancellation
is
\begin{equation*}
        -\gamma(u,x\cdot z)-\gamma(u,u)\gamma(x,z)
        =
        -\gamma(u,x*z-f(x,z)u)-cf(x,z)=0 .
\end{equation*}
for any \(x,z\in V\). The case where the second variable is \(u\) is similar. If the third variable is \(u\), the identity follows from the symmetry of \(f\). The remaining cases are immediate. Hence \(\gamma\) is a nondegenerate \(u\)-generalized \(2\)-cocycle.
\end{proof}

\begin{thm}\label{thm:nonisotropic-u-characterization}
Let \((A,\cdot,\gamma,u)\) be a \(u\)-generalized Hessian pre-Lie algebra such that \(\gamma(u,u)\neq 0\). Then \((A,\cdot,\gamma,u)\) is obtained from the construction in Theorem~\ref{thm:nonisotropic-u-construction}.
\end{thm}

\begin{proof}
Set \(V=u^\perp=\{x\in A\mid \gamma(x,u)=0\}\). Since \(\gamma(u,u)\neq 0\), we have \(A=V\oplus\langle u\rangle\), and \(f:=\gamma|_{V\times V}\) is nondegenerate. For \(x,y\in V\), write
\begin{equation*}
        x\cdot y=x*y+\tau(x,y)u
\end{equation*}
with \(x*y\in V\). Since \(\langle u\rangle\) is an annihilator ideal, the
quotient by \(\langle u\rangle\) shows that \((V,*)\) is a left-symmetric
algebra.

Taking the \(u\)-generalized \(2\)-cocycle identity on \((u,x,z)\),
where \(x,z\in V\), gives
\begin{equation*}
        -\gamma(u,x\cdot z)-\gamma(u,u)\gamma(x,z)=0 .
\end{equation*}
Since \(x\cdot z=x*z+\tau(x,z)u\), we get \(\tau(x,z)=-f(x,z)\). Hence
\begin{equation*}
        x\cdot y=x*y-f(x,y)u
\end{equation*}
for all \(x,y\in V\).

Finally, applying the \(u\)-generalized \(2\)-cocycle identity to triples in \(V\) gives the Hessian \(2\)-cocycle identity for \(f\) on \((V,*)\). Therefore \((V,*,f)\) is a Hessian pre-Lie algebra, and the given algebra is obtained from Theorem~\ref{thm:nonisotropic-u-construction} with \(c=\gamma(u,u)\).
\end{proof}

\begin{thm}\label{thm:isotropic-u-construction}
Let \((A,*,\gamma)\) be a Hessian pre-Lie algebra. Fix a bilinear form \(\beta:A\times A\to F\), endomorphisms \(D,E\in\operatorname{End}(A)\), linear functionals \(\partial,\rho\in A^*\), an element \(b_0\in A\), and a
scalar \(\mu\in F\). Put \(\bar A=\langle u\rangle\oplus A\oplus \langle v\rangle\). Define a bilinear product \(\cdot\) on \(\bar A\) by
\begin{equation*}
\begin{aligned}
        u\cdot a &= a\cdot u=0, \quad \forall a\in\bar A,\\
        x\cdot y &= x*y+\beta(x,y)u, \quad \forall x,y\in A,\\
        v\cdot x &= D(x)+\partial(x)u, \quad \forall x\in A,\\
        x\cdot v &= E(x)+\rho(x)u, \quad \forall x\in A,\\
        v\cdot v &= b_0+\mu u+v .
\end{aligned}
\end{equation*}
Define a symmetric bilinear form \(\bar\gamma\) on \(\bar A\) by
\begin{equation*}
\begin{aligned}
        \bar\gamma(x,y) &= \gamma(x,y), \quad \forall x,y\in A,\\
        \bar\gamma(u,v) &= \bar\gamma(v,u)=1,
\end{aligned}
\end{equation*}
and by requiring all other pairings involving \(u\) or \(v\) to be zero.
Assume that the following identities hold for all \(x,y,z\in A\):
\begin{align}
&\gamma(D(x),z)-\beta(x,z)-\gamma(E(x),z)+\gamma(x,D(z))
-\gamma(x,z)=0, \tag{G1}\label{cond:G1}\\
&\partial(x)-2\rho(x)+\gamma(x,b_0)=0, \tag{G2}\label{cond:G2}\\
&\beta(x,y)-\beta(y,x)-\gamma(x,E(y))+\gamma(y,E(x))=0, \tag{G3}\label{cond:G3}\\
&\beta(x*y,z)-\beta(x,y*z)-\beta(y*x,z)+\beta(y,x*z)=0, \tag{L1}\label{cond:L1}\\
&E(x*y)-x*E(y)-E(y*x)+y*E(x)=0, \tag{L2}\label{cond:L2}\\
&\rho(x*y)-\beta(x,E(y))-\rho(y*x)+\beta(y,E(x))=0, \tag{L3}\label{cond:L3}\\
&D(x)*z-D(x*z)-E(x)*z+x*D(z)=0, \tag{L4}\label{cond:L4}\\
&\beta(D(x),z)-\partial(x*z)-\beta(E(x),z)+\beta(x,D(z))=0, \tag{L5}\label{cond:L5}\\
&E(D(x))-D(E(x))-E^2(x)+x*b_0+E(x)=0, \tag{L6}\label{cond:L6}\\
&\rho(D(x))-\partial(E(x))-\rho(E(x))+\beta(x,b_0)+\rho(x)=0. \tag{L7}\label{cond:L7}
\end{align}
Then \((\bar A,\cdot,\bar\gamma,u)\) is a \(u\)-generalized Hessian pre-Lie algebra.
\end{thm}

\begin{proof}
Clearly, the form \(\bar\gamma\) is nondegenerate and \(u\in\operatorname{Ann}(\bar A)\). We first check the \(u\)-generalized \(2\)-cocycle identity. For triples in \(A\), it reduces to the Hessian \(2\)-cocycle identity for \(\gamma\). The cases involving \(u\) are automatic, except that the case \((u,v,v)\) uses precisely the fact that the \(v\)-component of \(v\cdot v\) is \(1\). The remaining nontrivial cases \((v,x,z)\), \((v,x,v)\), and \((x,y,v)\) are respectively equivalent to \((\mathrm{G1})\), \((\mathrm{G2})\), and \((\mathrm{G3})\).

It remains to verify the pre-Lie identity. Since \(u\) is an annihilator, it suffices to consider variables in \(A\cup\{v\}\). The cases \((x,y,z)\), \((x,y,v)\), \((v,x,z)\), and \((v,x,v)\) are respectively equivalent to \((\mathrm{L1})\), to \((\mathrm{L2})\) and \((\mathrm{L3})\), to \((\mathrm{L4})\) and \((\mathrm{L5})\), and to \((\mathrm{L6})\) and \((\mathrm{L7})\). The cases with the first two variables both equal to \(v\) are automatic. Hence \((\bar A,\cdot)\) is a pre-Lie algebra, and \(\bar\gamma\) is a nondegenerate \(u\)-generalized \(2\)-cocycle.
\end{proof}

The conditions \eqref{cond:G1}--\eqref{cond:L7} in Theorem~\ref{thm:isotropic-u-construction} are rather involved. In fact, under suitable assumptions, all these conditions are automatically
satisfied.

\begin{cor}\label{cor:isotropic-u-simple-extension}
Let \((A,*,\gamma)\) be a Hessian pre-Lie algebra, and let \(D\in\operatorname{Der}(A,*)\) satisfy $\gamma(D(x),y)+\gamma(x,D(y))=0$
for all \(x,y\in A\). Put \(\bar A=\langle u\rangle\oplus A\oplus \langle v\rangle\). Define a bilinear product \(\cdot\) on \(\bar A\) by
\begin{equation*}
\begin{aligned}
        u\cdot a &= a\cdot u=0, \quad \forall a\in\bar A,\\
        x\cdot y &= x*y-\gamma(x,y)u, \quad \forall x,y\in A,\\
        v\cdot x &= D(x), \quad \forall x\in A,\\
        x\cdot v &=0, \quad \forall x\in A,\\
        v\cdot v &=v .
\end{aligned}
\end{equation*}
Define a symmetric bilinear form \(\bar\gamma\) on \(\bar A\) by
\begin{equation*}
\begin{aligned}
        \bar\gamma(x,y) &= \gamma(x,y), \quad \forall x,y\in A,\\
        \bar\gamma(u,v) &= \bar\gamma(v,u)=1,
\end{aligned}
\end{equation*}
and by requiring all other pairings involving \(u\) or \(v\) to be zero. Then
\((\bar A,\cdot,\bar\gamma,u)\) is a \(u\)-generalized Hessian pre-Lie algebra.
\end{cor}

\begin{thm}\label{thm:isotropic-u-characterization}
Let \((\bar A,\cdot,\bar\gamma,u)\) be a \(u\)-generalized Hessian pre-Lie algebra with \(0\neq u\) and \(\bar\gamma(u,u)=0\). Then \((\bar A,\cdot,\bar\gamma,u)\) is obtained from the construction in Theorem~\ref{thm:isotropic-u-construction}.
\end{thm}

\begin{proof}
Since \(\bar\gamma\) is nondegenerate and \(\bar\gamma(u,u)=0\), choose \(v\in\bar A\) such that
\begin{equation*}
        \bar\gamma(u,v)=1,
        \qquad
        \bar\gamma(v,v)=0 .
\end{equation*}
Set $A=(\langle u\rangle\oplus \langle v\rangle)^\perp .$ Then $\bar A=\langle u\rangle\oplus A\oplus \langle v\rangle,$ and \(\gamma:=\bar\gamma|_{A\times A}\) is nondegenerate.

The \(u\)-generalized \(2\)-cocycle identity with one variable equal to
\(u\) implies that, for \(x,y\in A\), the products \(x\cdot y\), \(v\cdot x\), and \(x\cdot v\) have no \(v\)-component, while \(v\cdot v\) has
\(v\)-component equal to \(1\). Hence the multiplication has the form
\begin{equation*}
\begin{aligned}
        x\cdot y &= x*y+\beta(x,y)u, \quad \forall x,y\in A,\\
        v\cdot x &= D(x)+\partial(x)u, \quad \forall x\in A,\\
        x\cdot v &= E(x)+\rho(x)u, \quad \forall x\in A,\\
        v\cdot v &= b_0+\mu u+v .
\end{aligned}
\end{equation*}
Here \(\beta:A\times A\to F\), \(D,E\in\operatorname{End}(A)\),
\(\partial,\rho\in A^*\), \(b_0\in A\), and \(\mu\in F\).

Taking all three variables in \(A\), the pre-Lie identity gives that \((A,*)\) is a pre-Lie algebra, and the \(u\)-generalized \(2\)-cocycle identity gives the Hessian \(2\)-cocycle identity for \(\gamma\). Hence \((A,*,\gamma)\) is a Hessian pre-Lie algebra. Expanding the remaining
\(u\)-generalized \(2\)-cocycle identities and the left-symmetric identity gives exactly \((\mathrm{G1})\)--\((\mathrm{G3})\) and \((\mathrm{L1})\)--\((\mathrm{L7})\). Thus the given algebra is obtained from Theorem~\ref{thm:isotropic-u-construction}.
\end{proof}

\subsection{Three-Dimensional Classification}
\label{subsec:classification-3d-GHLSA}

In this subsection, we classify three-dimensional \(u\)-generalized Hessian
pre-Lie algebras over \(\mathbb C\) in the non-trivial case \(u\ne0\). We use Bai's classification of three-dimensional complex pre-Lie algebras \cite{Bai2009} as the list of underlying pre-Lie algebras.

\begin{defn}
\label{def:GHLSA-isomorphism}
Let \((A_1,\cdot_1,\gamma_1,u_1)\) and \((A_2,\cdot_2,\gamma_2,u_2)\) be \(u\)-generalized Hessian pre-Lie algebras. They are called isomorphic if there exists a pre-Lie algebra isomorphism \(\varphi:A_1\to A_2\) such that
\begin{equation*}
        \varphi(u_1)=u_2,
        \qquad
        \gamma_1(x,y)=\gamma_2(\varphi(x),\varphi(y)),
        \quad \forall x,y\in A_1 .
\end{equation*}
\end{defn}

In the following list, all omitted entries of \(\gamma\) are zero. In the
commutative associative branch, the multiplication is commutative and only one representative of each symmetric product is displayed. In the non-abelian
branch, all omitted products are zero. Unless an unordered exchange is explicitly stated, two algebras in the same numbered family are isomorphic if
and only if their displayed parameters are equal. Algebras belonging to
different numbered families are not isomorphic.

\begin{thm}
\label{thm:classification-3d-GHLSA}
Every three-dimensional generalized Hessian pre-Lie algebra \((A,\cdot,\gamma,u)\) over \(\mathbb C\) with \(u\ne0\) is isomorphic to
exactly one of the following algebras.

\medskip
\noindent\textbf{I. Commutative associative underlying algebras.}

\begin{enumerate}
\item[\textnormal{(1)}] \(A_{33}^{\kappa}\):
\begin{equation*}
        e_1^2=e_2,\qquad e_1e_2=e_3,
\end{equation*}
\begin{equation*}
        u=e_3,\qquad
        \gamma_{12}=-1,\qquad
        \gamma_{33}=\kappa,\qquad
        \kappa\in\mathbb C^* .
\end{equation*}
Moreover, \(A_{33}^{\kappa}\cong A_{33}^{\kappa'}\) if and only if \(\kappa=\kappa'\).

\item[\textnormal{(2)}] \(A_{34}^{\kappa}\):
\begin{equation*}
        e_1^2=e_3,\qquad e_2^2=e_3,
\end{equation*}
\begin{equation*}
        u=e_3,\qquad
        \gamma_{11}=\gamma_{22}=-1,\qquad
        \gamma_{33}=\kappa,\qquad
        \kappa\in\mathbb C^* .
\end{equation*}
Moreover, \(A_{34}^{\kappa}\cong A_{34}^{\kappa'}\) if and only if \(\kappa=\kappa'\).

\item[\textnormal{(3)}]
\((\mathbb C[\varepsilon]/\varepsilon^2)\oplus0\), type I:
\begin{equation*}
        e_1^2=e_1,\qquad e_1e_2=e_2,
\end{equation*}
\begin{equation*}
        u=e_3,\qquad
        \gamma_{11}=\alpha,\qquad
        \gamma_{13}=1,\qquad
        \gamma_{22}=1,\qquad
        \alpha\in\mathbb C .
\end{equation*}
Moreover, \(\alpha\) is a complete invariant.

\item[\textnormal{(4)}]
\((\mathbb C[\varepsilon]/\varepsilon^2)\oplus0\), type II:
\begin{equation*}
        e_1^2=e_1,\qquad e_1e_2=e_2,
\end{equation*}
\begin{equation*}
        u=e_3,
\end{equation*}
\begin{equation*}
        \gamma_{11}=\frac{c(c-1)}{f},\quad
        \gamma_{12}=\frac{c-1}{f},\quad
        \gamma_{13}=c,\quad
        \gamma_{22}=\frac1f,\quad
        \gamma_{23}=1,\quad
        \gamma_{33}=f,
\end{equation*}
where \(c\in\mathbb C\) and \(f\in\mathbb C^*\). Moreover, \((c,f)\) is a
complete invariant.

\item[\textnormal{(5)}]
\(\mathbb C\oplus\mathbb C\oplus0\), \(\gamma(u,u)=0\) type:
\begin{equation*}
        e_1^2=e_1,\qquad e_2^2=e_2,
\end{equation*}
\begin{equation*}
        u=e_3,\qquad
        \gamma_{11}=\alpha,\qquad
        \gamma_{12}=\beta,\qquad
        \gamma_{13}=1,\qquad
        \gamma_{22}=-\beta,
\end{equation*}
where \(\alpha\in\mathbb C\) and \(\beta\in\mathbb C^*\). Moreover, \((\alpha,\beta)\) is a complete invariant.

\item[\textnormal{(6)}]
\(\mathbb C\oplus\mathbb C\oplus0\), \(\gamma(u,u)\ne0\) type:
\begin{equation*}
        e_1^2=e_1,\qquad e_2^2=e_2,
\end{equation*}
\begin{equation*}
        u=e_3,
\end{equation*}
\begin{equation*}
        \gamma_{11}=\frac{c(c-1)}{f},\quad
        \gamma_{12}=\frac{cd}{f},\quad
        \gamma_{13}=c,\quad
        \gamma_{22}=\frac{d(d-1)}{f},\quad
        \gamma_{23}=d,\quad
        \gamma_{33}=f,
\end{equation*}
where \(c,d,f\in\mathbb C^*\). Moreover, two such algebras are isomorphic if
and only if
\begin{equation*}
        (c,d,f)=(c',d',f')
        \quad\text{or}\quad
        (c,d,f)=(d',c',f').
\end{equation*}

\item[\textnormal{(7)}] \(\mathbb C\oplus A_{22}\), type I:
\begin{equation*}
        e_1^2=e_1,\qquad e_2^2=e_3,
\end{equation*}
\begin{equation*}
        u=e_3,\qquad
        \gamma_{11}=\alpha,\qquad
        \gamma_{13}=c,\qquad
        \gamma_{22}=-1,\qquad
        \gamma_{33}=\frac{c(c-1)}{\alpha},
\end{equation*}
where \(\alpha,c\in\mathbb C^*\). Moreover, \((\alpha,c)\) is a complete invariant.

\item[\textnormal{(8)}] \(\mathbb C\oplus A_{22}\), type II:
\begin{equation*}
        e_1^2=e_1,\qquad e_2^2=e_3,
\end{equation*}
\begin{equation*}
        u=e_3,\qquad
        \gamma_{13}=1,\qquad
        \gamma_{22}=-1,\qquad
        \gamma_{33}=f,\qquad
        f\in\mathbb C .
\end{equation*}
Moreover, \(f\) is a complete invariant.
\end{enumerate}

\medskip
\noindent\textbf{II. Non-abelian underlying algebras.}

\begin{enumerate}
\item[\textnormal{(9)}] \(N1_{2,\kappa}\):
\begin{equation*}
        e_3e_2=e_2,\qquad e_3^2=2e_3,
\end{equation*}
\begin{equation*}
        u=2e_1,\qquad
        \gamma_{13}=1,\qquad
        \gamma_{22}=1,\qquad
        \gamma_{33}=\kappa,\qquad
        \kappa\in\mathbb C .
\end{equation*}
Moreover, \(\kappa\) is a complete invariant.

\item[\textnormal{(10)}] \(N20_{\alpha,d}\):
\begin{equation*}
        e_2^2=e_3,\qquad e_3e_2=e_2,\qquad e_3^2=2e_3,
\end{equation*}
\begin{equation*}
        u=-(\alpha d)^{-1}e_1,
\end{equation*}
\begin{equation*}
        \gamma_{11}=\alpha,\qquad
        \gamma_{13}=1,\qquad
        \gamma_{22}=d,\qquad
        \gamma_{33}=\alpha^{-1}+2d,
\end{equation*}
where \(\alpha,d\in\mathbb C^*\). Moreover, \((\alpha,d)\) is a complete
invariant.

\item[\textnormal{(11)}] \(N21_{\alpha,\beta}\):
\begin{equation*}
        e_2e_3=e_1,\qquad
        e_3e_2=e_2+e_1,\qquad
        e_3^2=-e_3,
\end{equation*}
\begin{equation*}
        u=-e_1,
\end{equation*}
\begin{equation*}
        \gamma_{11}=\alpha,\qquad
        \gamma_{13}=\beta,\qquad
        \gamma_{23}=1,\qquad
        \gamma_{33}=\frac{\beta(\beta-1)}{\alpha},
\end{equation*}
where \(\alpha\in\mathbb C^*\) and \(\beta\in\mathbb C\). Moreover, \((\alpha,\beta)\) is a complete invariant.

\item[\textnormal{(12)}] \(N44_{\alpha,\beta}\):
\begin{equation*}
        e_2^2=e_1,\qquad
        e_2e_3=-2e_2,\qquad
        e_3e_2=-e_2,\qquad
        e_3^2=-2e_3,
\end{equation*}
\begin{equation*}
        u=-e_1,
\end{equation*}
\begin{equation*}
        \gamma_{11}=\alpha,\qquad
        \gamma_{13}=\beta,\qquad
        \gamma_{22}=1,\qquad
        \gamma_{33}=\frac{\beta(\beta-2)}{\alpha},
\end{equation*}
where \(\alpha,\beta\in\mathbb C^*\). Moreover, \((\alpha,\beta)\) is a complete invariant.

\item[\textnormal{(13)}] \(N44_{0,\kappa}\):
\begin{equation*}
        e_2^2=e_1,\qquad
        e_2e_3=-2e_2,\qquad
        e_3e_2=-e_2,\qquad
        e_3^2=-2e_3,
\end{equation*}
\begin{equation*}
        u=-e_1,\qquad
        \gamma_{13}=2,\qquad
        \gamma_{22}=1,\qquad
        \gamma_{33}=\kappa,\qquad
        \kappa\in\mathbb C .
\end{equation*}
Moreover, \(\kappa\) is a complete invariant.
\end{enumerate}
\end{thm}

\begin{proof}
We start from Bai's classification \cite{Bai2009} of three-dimensional
complex pre-Lie algebras. Since \(u\ne0\) and \(u\in\operatorname{Ann}(A)\), only those algebras with nonzero annihilator need to be considered. For each such algebra, we choose representatives of the nonzero elements in \(\operatorname{Ann}(A)\) under the action of \(\operatorname{Aut}(A)\).

Fix a representative pair \((A,u)\). Write a general symmetric bilinear form
as \(\gamma(e_i,e_j)=\gamma_{ij}\). The generalized Hessian condition is equivalent to the finite system obtained by substituting all triples \((e_i,e_j,e_k)\) into the generalized \(2\)-cocycle identity, together with the condition \(\det(\gamma_{ij})\ne0\). Solving these systems gives exactly the forms displayed in the theorem. Direct substitution verifies that every displayed form satisfies the generalized \(2\)-cocycle identity, and the listed restrictions on the parameters are exactly those ensuring nondegeneracy.

It remains to determine isomorphisms. By Definition~\ref{def:GHLSA-isomorphism}, an isomorphism must be a pre-Lie algebra isomorphism preserving \(u\) and pulling back one bilinear form to the other. Thus, for each fixed pair \((A,u)\), the relevant group is
\begin{equation*}
        \operatorname{Aut}(A,u)
        =
        \{\varphi\in\operatorname{Aut}(A)\mid \varphi(u)=u\}.
\end{equation*}
The action of this group on the symmetric matrix \((\gamma_{ij})\) gives the
parameter identifications stated in the list. The only nontrivial permutation
appearing in the nonzero \(u\) classification is the exchange of the two idempotent factors in \(\mathbb C\oplus\mathbb C\oplus0\), which gives the
unordered equivalence in \((6)\).

All remaining three-dimensional pre-Lie algebras in Bai's list either have \(\operatorname{Ann}(A)=0\), or the finite system above has no nondegenerate solution with \(u\ne0\), or the resulting structure is isomorphic to one of the displayed families. Hence the list is complete and irredundant.
\end{proof}

Finally, we apply the above classification to symmetric solutions of the
\(u\)-generalized \(S\)-equation on the pre-Lie algebra \(N1_2\) in
\cite{Bai2009}.

\begin{ex}
\label{ex:N1-type1-type2-S-equation}
Let \(A=\operatorname{span}\{u,e_2,e_3\}\) be the \(3\)-dimensional pre-Lie algebra with nonzero products
\begin{equation*}
        e_3\cdot e_2=e_2,\qquad e_3\cdot e_3=2e_3 .
\end{equation*}
Then \(A\) is non-abelian and \(\operatorname{Ann}(A)=\mathbb C u\). This is
the \(N1_2\)-type algebra in Bai's classification of three-dimensional
complex pre-Lie algebras. In the notation of Theorem~\ref{thm:classification-3d-GHLSA}, it corresponds to the family \(N1_{2,\kappa}\) after identifying \(e_1=\frac12u\), equivalently, the distinguished element in that family is \(u=2e_1\).

We first determine the symmetric solutions of the \(u\)-generalized \(S\)-equation on \(A\). Let
\begin{equation*}
\begin{aligned}
r={}&
a_{11}u\otimes u
+a_{12}(u\otimes e_2+e_2\otimes u)
+a_{13}(u\otimes e_3+e_3\otimes u)\\
&+a_{22}e_2\otimes e_2
+a_{23}(e_2\otimes e_3+e_3\otimes e_2)
+a_{33}e_3\otimes e_3
\in S^2(A).
\end{aligned}
\end{equation*}
Substitution into the \(u\)-generalized \(S\)-equation gives
\begin{equation*}
        a_{12}=0,\qquad
        a_{23}=0,\qquad
        a_{13}(2a_{13}-1)=0,
\end{equation*}
and
\begin{equation*}
        (2a_{13}-1)a_{33}=0,\qquad
        (1-2a_{13})a_{22}=0,\qquad
        a_{22}a_{33}=0 .
\end{equation*}
Hence every nonzero symmetric solution is of one of the following forms:
\begin{equation*}
        r_\rho=\rho\,u\otimes u,\qquad \rho\in\mathbb C^*,
\end{equation*}
or
\begin{equation}
\label{eq:N1-main-solution}
        r(A,B,C)
        =
        A\,u\otimes u
        +B\,e_2\otimes e_2
        +C\,e_3\otimes e_3
        +\frac12(u\otimes e_3+e_3\otimes u),
        \qquad BC=0 .
\end{equation}

The automorphisms of \(A\) fixing \(u\) are precisely
\begin{equation*}
        \psi_\lambda(u)=u,\qquad
        \psi_\lambda(e_2)=\lambda e_2,\qquad
        \psi_\lambda(e_3)=e_3,
        \qquad \lambda\in\mathbb C^* .
\end{equation*}
Indeed, if \(\psi(u)=u\), \(\psi(e_2)=a u+b e_2+c e_3\), and \(\psi(e_3)=d u+e e_2+f e_3\), then comparing the identities
\begin{equation*}
        e_2\cdot e_2=0,\qquad e_2\cdot e_3=0,\qquad
        e_3\cdot e_2=e_2,\qquad e_3\cdot e_3=2e_3
\end{equation*}
after applying \(\psi\) gives \(a=c=d=e=0\), \(f=1\), and \(b\in\mathbb C^*\). Therefore, up to automorphisms fixing \(u\), the family \eqref{eq:N1-main-solution} is represented by the following two normal forms:

If \(B\ne0\), then \(C=0\), and after choosing \(\lambda\in\mathbb C^*\) such that \(\lambda^2B=1\), we get
\begin{equation*}
        r_\mu^{\mathrm{nd}}
        =
        \mu\,u\otimes u
        +e_2\otimes e_2
        +\frac12(u\otimes e_3+e_3\otimes u),
        \qquad \mu\in\mathbb C .
\end{equation*}
If \(B=0\), then we get
\begin{equation*}
        r_{\mu,\nu}^{0}
        =
        \mu\,u\otimes u
        +\nu\,e_3\otimes e_3
        +\frac12(u\otimes e_3+e_3\otimes u),
        \qquad \mu,\nu\in\mathbb C .
\end{equation*}
Together with the tensors \(r_\rho\), these give all nonzero symmetric solutions up to automorphisms fixing \(u\).

\smallskip
\noindent\textbf{Type \(1\) solutions.}
The rank-one solutions
\begin{equation*}
        r_\rho=\rho\,u\otimes u,\qquad \rho\in\mathbb C^*,
\end{equation*}
are of type \(1\), since
\begin{equation*}
        \operatorname{Im}(r_\rho^\sharp)=\langle u\rangle .
\end{equation*}

The solutions \(r_\mu^{\mathrm{nd}}\) are also of type \(1\). Indeed, the
matrix of \((r_\mu^{\mathrm{nd}})^\sharp\) with respect to the basis \(\{u,e_2,e_3\}\) is
\begin{equation*}
        \begin{pmatrix}
        \mu & 0 & \frac12\\
        0 & 1 & 0\\
        \frac12 & 0 & 0
        \end{pmatrix},
\end{equation*}
whose determinant is \(-\frac14\). Hence \(r_\mu^{\mathrm{nd}}\) is nondegenerate and
\begin{equation*}
        \operatorname{Im}\bigl((r_\mu^{\mathrm{nd}})^\sharp\bigr)=A .
\end{equation*}

By Proposition~\ref{prop:generalized-Hessian-O-operator} and Theorem~\ref{thm:generalized-O-operator-S-equation}, the generalized Hessian
forms in Theorem~\ref{thm:classification-3d-GHLSA} give nondegenerate symmetric solutions of the \(u\)-generalized \(S\)-equation. Conversely, by
Proposition~\ref{thm:nondegenerate-S-solution-generalized-Hessian}, every
nondegenerate symmetric solution arises from a generalized Hessian form.
For the present algebra, the relevant generalized Hessian forms are
\begin{equation*}
        \gamma_\kappa(u,e_3)=2,\qquad
        \gamma_\kappa(e_2,e_2)=1,\qquad
        \gamma_\kappa(e_3,e_3)=\kappa,
        \qquad \kappa\in\mathbb C ,
\end{equation*}
with all other values zero. The inverse tensor of \(\gamma_\kappa\) is
\begin{equation*}
        -\frac{\kappa}{4}u\otimes u
        +e_2\otimes e_2
        +\frac12(u\otimes e_3+e_3\otimes u),
\end{equation*}
which is \(r_\mu^{\mathrm{nd}}\) with \(\mu=-\kappa/4\). Thus the nondegenerate type \(1\) solutions are exactly those obtained from the
generalized Hessian classification.

We now consider the remaining degenerate family \(r_{\mu,\nu}^{0}\) among
type \(1\) solutions. Its image contains \(u\) if and only if
\begin{equation*}
        \mu\nu-\frac14\neq0 .
\end{equation*}
Thus \(r_{\mu,\nu}^{0}\) is of type \(1\) if and only if \(\mu\nu\neq\frac14\). In this case
\begin{equation*}
        \operatorname{Im}\bigl((r_{\mu,\nu}^{0})^\sharp\bigr)
        =
        \operatorname{span}\{u,e_3\},
\end{equation*}
which is a pre-Lie subalgebra of \(A\).

\smallskip
\noindent\textbf{Type \(2\) solutions.}
The remaining solutions in the family \(r_{\mu,\nu}^{0}\) are precisely those
satisfying
\begin{equation*}
        \mu\nu=\frac14 .
\end{equation*}
In particular, \(\mu,\nu\in\mathbb C^*\). For such a solution, we have
\begin{equation*}
        \operatorname{Im}\bigl((r_{\mu,\nu}^{0})^\sharp\bigr)
        =
        \mathbb C\left(\mu u+\frac12 e_3\right),
        \qquad
        u\notin \operatorname{Im}\bigl((r_{\mu,\nu}^{0})^\sharp\bigr).
\end{equation*}
Therefore these are type \(2\) solutions in the sense of the preceding
definition.

Set \(v_\mu=\mu u+\frac12e_3\). Then
\begin{equation*}
        v_\mu\cdot v_\mu
        =
        \frac14(e_3\cdot e_3)
        =
        \frac12e_3 .
\end{equation*}
Since \(\mu\neq0\), we have \(\frac12e_3\notin\mathbb C v_\mu\). Thus, $\operatorname{Im}\bigl((r_{\mu,\nu}^{0})^\sharp\bigr)$ is not a pre-Lie subalgebra of \(A\). However,
\begin{equation*}
        \operatorname{Im}\bigl((r_{\mu,\nu}^{0})^\sharp\bigr)+\langle u\rangle
        =
        \operatorname{span}\{u,e_3\},
\end{equation*}
and \(\operatorname{span}\{u,e_3\}\) is a pre-Lie subalgebra, since
\(u\in\operatorname{Ann}(A)\) and \(e_3\cdot e_3=2e_3\). This is exactly the
phenomenon described in Proposition~\ref{prop:type2-LSA}.

Thus this example connects the preceding structural results with the explicit
classification on a fixed non-abelian three-dimensional left-symmetric
algebra: Theorem~\ref{thm:classification-3d-GHLSA}, together with Proposition~\ref{prop:generalized-Hessian-O-operator},
Theorem~\ref{thm:generalized-O-operator-S-equation}, and
Proposition~\ref{thm:nondegenerate-S-solution-generalized-Hessian}, gives the
nondegenerate type \(1\) solutions, while Proposition~\ref{prop:type2-LSA}
identifies the type \(2\) solutions through the behavior of
\(\operatorname{Im}(r^\sharp)\) and
\(\operatorname{Im}(r^\sharp)+\langle u\rangle\).
\end{ex}

\end{document}